\documentclass[10pt,twocolumn,twoside]{IEEEtran}
\usepackage{amsmath,amsfonts}
\usepackage{algorithmic}
\usepackage{algorithm}
\usepackage{array}
\usepackage{booktabs}
\usepackage{lipsum}
\usepackage{ragged2e}
\usepackage{textcomp}
\usepackage{stfloats}
\usepackage{url}
\usepackage{verbatim}
\usepackage{graphicx}
\usepackage{relsize}
\usepackage{cite}
\usepackage{subfigure}
\usepackage{bm}
\usepackage{amsthm}
\usepackage{color}
\usepackage{enumerate}
\usepackage{amssymb}
\usepackage{mathrsfs}
\usepackage{balance}
\usepackage[hidelinks]{hyperref}
\usepackage{mathtools}
\usepackage{multirow}
\usepackage{siunitx}
\newcommand{\R}{\mathbb{R}}
\DeclareMathOperator*{\argmin}{arg\,min}
\newcommand{\E}[1]{\mathbb{E}\!\left[\,#1\,\right]}
\newcommand{\EC}[2]{\mathbb{E}\!\left[\,#1 \,\middle|\, #2\,\right]}

\newcommand{\VarC}[2]{\operatorname{Var}\!\left[\,#1 \,\middle|\, #2\,\right]}

\newcommand{\CovC}[2]{\operatorname{Cov}\!\left[\,#1 \,\middle|\, #2\,\right]}
\newcommand{\Prob}[1]{\mathbb{P}\!\left\{#1\right\}}
\newcommand{\Exp}[1]{\exp\!\left(#1\right)}
\DeclareMathOperator{\diag}{diag}
\theoremstyle{plain}
\newtheorem{theorem}{Theorem}[section]
\newtheorem{lemma}{Lemma}[section]

\theoremstyle{definition}
\newtheorem{definition}{Definition}[section]

\theoremstyle{remark}

\begin{document}
\title{Stochastic Control Barrier Functions under State Estimation: From Euclidean Space to Lie Groups
}
\author{Ruoyu Lin and Magnus Egerstedt
\thanks{This work was supported by the U.S. Army Research Lab through ARL DCIST CRA W911NF-17-2-0181}
\thanks{Ruoyu Lin and Magnus Egerstedt are with the Department of Electrical Engineering and Computer Science, University of California, Irvine, CA 92697, USA. Email: {\tt\small \{\href{mailto:rlin10@uci.edu}{\tt\small rlin10}, \href{mailto:magnus@uci.edu}{\tt\small magnus\}@uci.edu}}}
}
\maketitle
\begin{abstract}
Ensuring safety for autonomous systems under uncertainty remains challenging, particularly when safety of the true state is required despite the true state not being fully known. Control barrier functions (CBFs) have become widely adopted as safety filters. However, standard CBF formulations do not explicitly account for state estimation uncertainty and its propagation, especially for stochastic systems evolving on manifolds. In this paper, we propose a safety-critical control framework with a provable bound on the finite-time safety probability for stochastic systems under noisy state information. The proposed framework explicitly incorporates the uncertainty arising from both process and measurement noise, and synthesizes controllers that adapt to the level of uncertainty. The framework admits closed-form solutions in linear settings, and experimental results demonstrate its effectiveness on systems whose state spaces range from Euclidean space to Lie groups.
\end{abstract}

\section{Introduction} \label{Section:Introduction}
\subsection{Background and Motivation} \label{Section:Intro_BM}
Safety is a fundamental requirement for autonomous systems deployed in direct interaction with humans and complex environments, from robotic manipulators performing minimally invasive surgery to autonomous vehicles navigating dense urban traffic, where safety violations can lead to catastrophic consequences \cite{pek2020fail,alemzadeh2016adverse}. Control barrier functions (CBFs) have emerged as a widely adopted safety filter to ensure forward invariance, i.e., keeping the state in a safe set, by minimally modifying nominal control policies \cite{ames2016control}. However, the standard CBF formulation assumes access to ground-truth states of the underlying dynamical systems.

In practice, autonomous systems inevitably operate under uncertainty arising from noisy or incomplete sensing, unmodeled dynamics, and stochastic disturbances. A common approach is to enforce safety under worst-case bounded deterministic disturbances, but this often results in overly conservative behavior. Notably, the level of uncertainty depends on operating conditions. For example, degraded sensing in fog increases observation uncertainty, and slippery or uneven terrain increases uncertainty in the system dynamics, both calling for more cautious actions. Conversely, when state estimates are reliable, worst-case conservatism becomes unnecessarily restrictive. This raises a question: How can safety of the true state still be ensured when the true state cannot be fully known, while the control policy adapts to the level of uncertainty?

Beyond uncertainty, an additional challenge stems from the geometry of the state manifold. In practice, the states of dynamical systems evolve on non-Euclidean spaces whenever rotational degrees of freedom are present, which has motivated the development of geometric control techniques for global property assessment, structural preservation, or vibrational stabilization (see, e.g., \cite{sussmann1973orbits,brockett1983asymptotic,bullo1999tracking,taha2020vibrational}). From an application perspective, when the size of a robot is comparable to the scale of its environment, its shape becomes critical, particularly in narrow environments (see, e.g., \cite{lavalle2006planning}). For example, consider an unmanned aerial vehicle (UAV) maneuvering at high speed through a cluttered environment. In such settings, considering only the position of the UAV or approximating the UAV as a point mass is insufficient to design a safety filter, since collision avoidance depends on the UAV's pose, which evolves on a Lie group. Furthermore, for stochastic systems, the importance of the geometry of the state manifold becomes more pronounced because it governs how uncertainty influences the state evolution (see, e.g., \cite{brockett1997notes,thrun2000real,long2013banana,marques2025lies}), as illustrated in Fig.~\ref{fig:Illustration}. This is different from the scenario where the state of a dynamical system is in Euclidean space and the noise is typically additive.

\begin{figure}[t]
\centering
\includegraphics[scale=0.2]{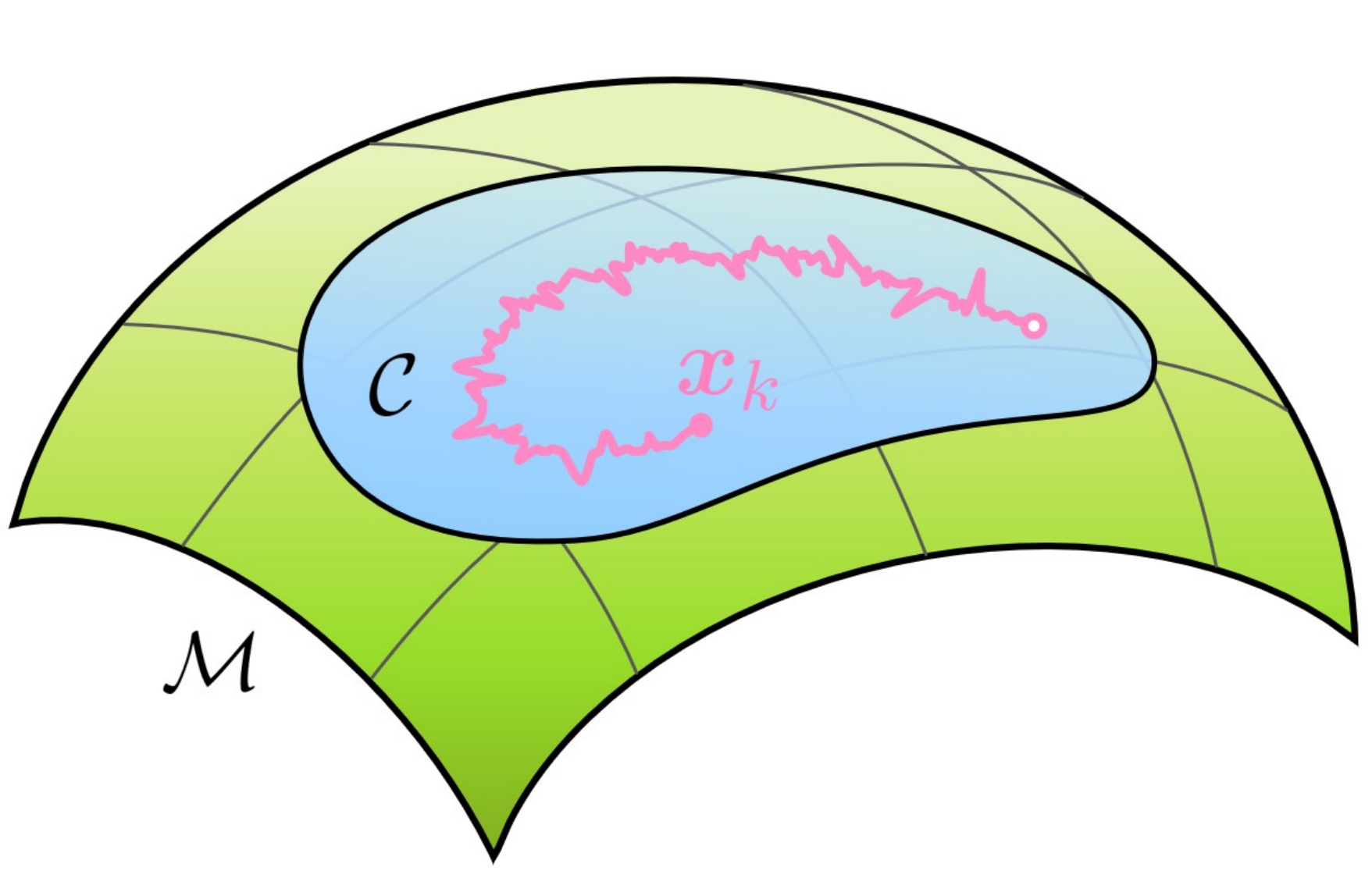}
\caption{Trajectory (pink) of a stochastic system evolving on a safe submanifold (blue) of the state manifold (green).}
\label{fig:Illustration}
\end{figure}
\subsection{Related Work} \label{Section:Intro_RW}
\subsubsection{Safety for Stochastic Systems}
Stochastic control barrier functions (SCBFs) have been introduced for safety-critical control of stochastic systems and it has been proved that almost-sure safety cannot be achieved by SCBFs under stochastic disturbances with unbounded support over infinite horizons \cite{so2023almost}. Based on \cite{kushner2003finite}, an SCBF-based framework for finite-time safety verification and control of stochastic systems is proposed in \cite{santoyo2019verification} and later extended in \cite{santoyo2021barrier}. However, this framework synthesizes the controller offline under polynomial assumptions on the CBFs, system dynamics, and controller, rather than through online optimization, which limits its applicability to general CBFs and adaptivity to the environment.

In contrast, \cite{cosner2023robust} proposes an online optimization-based control synthesis framework with finite-time safety guarantees for stochastic systems. However, it requires the CBF to be upper-bounded in order to construct probabilistic safety guarantees. Subsequently,  \cite{cosner2024bounding} relaxes this requirement by assuming stronger structural information on the system (e.g., a bounded difference between the true and predictable CBF values), yielding tighter theoretical bounds. However, the control synthesis framework in \cite{cosner2023robust,cosner2024bounding} does not account for the second-order moment information of CBFs so the resulting controllers do not adapt to the level of uncertainty. This may result in overly aggressive behavior when the level of uncertainty is relatively high, which will be demonstrated in Section~\ref{Section:motion planning}. In addition,
\cite{mestres2025probabilistic} proposes a probabilistic CBF framework, where the CBF-based constraint is reformulated as a chance constraint for stochastic systems, which is then converted into a computable deterministic surrogate constraint via probabilistic inequalities (e.g., Cantelli's inequality) for control synthesis. 

Another line of work is on safety verification. \cite{prajna2004stochastic} proposes a barrier-certificate framework for stochastic systems with provable bounds on the probability of reaching unsafe sets. However, the barrier functions, the safe sets, and the system dynamics are assumed to be polynomial for tractable computation. Recently, \cite{liu2024safety,liu2025safety} investigates finite-time safety verification of stochastic systems by reducing the problem to deterministic verification on an eroded safe set constructed via Minkowski difference. However, even when an analytic function (e.g., a CBF) defining the original safe set is given, constructing such eroded sets may be challenging in complex environments or high-dimensional state spaces, similar to the difficulty encountered in analytic free configuration-space construction in motion planning \cite{lavalle2006planning}.

It is worth mentioning that all the aforementioned works assume access to ground-truth states of dynamical systems without considering state estimation, and they do not consider systems whose states evolve on manifolds.

\subsubsection{CBFs on Manifolds}
Geometric control barrier functions (GCBFs) are introduced in \cite{wu2015safety} for dynamical systems whose states evolve on manifolds, and extended in \cite{wu2016safety} to handle time-varying safety constraints, enabling coordinate-free safety-critical control for fully actuated systems. Recently, a general framework of GCBFs on bundles is developed in \cite{de2025bundles}, which also introduces a backstepping-based synthesis method that constructs safety guarantees on the tangent bundle from configuration-space constraints for a class of underactuated systems. In addition, \cite{letti2025safety} proposes energy-augmented control barrier functions for safety-critical control on Lie groups, and validates their use for configuration-based obstacle avoidance and directional kinetic energy limiting. However, the aforementioned works do not consider stochasticity.
\subsection{Contributions} \label{Section:Intro_C}
In this paper, we aim to develop a general safety-critical control framework, which can deal with stochastic systems evolving either in Euclidean space or on manifolds while considering state estimation. The main contributions of this paper are as follows:
\begin{itemize}
\item
We propose a safety-critical control framework for stochastic systems where the ground-truth state is unavailable and decisions must be made from noisy state information. The proposed framework explicitly incorporates state estimation uncertainty arising jointly from dynamics noise and measurement noise. We also provide an offline computable theoretical bound on the finite-time safety probability, based on martingale theory, under mild assumptions. The CBF need not be upper bounded, convex, concave, or polynomial, and no bound is assumed on the difference between the true and predicted CBF values. Furthermore, the framework supports online, optimization-based control synthesis and yields behavior that adapts to the level of uncertainty.
\item 
For linear systems with affine CBFs, we derive closed-form expressions of the proposed framework, which admits a quadratic program (QP) for control synthesis. Then, we extend the expectation-based CBF \cite{cosner2023robust} and the probabilistic CBF \cite{mestres2025probabilistic} to the formulations under state estimation uncertainty, which also support CBF-QP to synthesize controllers. Moreover, the proposed framework offers a new perspective on motion planning, complementary to existing approaches such as sampling-based planning and diffusion-model-based planning, e.g., \cite{lavalle1998rapidly,kavraki2002probabilistic,janner2022planning}. We compare the proposed framework with the extended expectation-based CBF and probabilistic CBF under state estimation uncertainty for motion planning. The experimental results demonstrate that, within a given time, the proposed framework generates significantly more collision-free trajectories that successfully reach the goal than the other two methods. We further provide theoretical explanations of this advantage.
\item 
The proposed framework naturally accommodates general nonlinear stochastic systems whose states evolve on manifolds. In this work, we focus on the special Euclidean group $\mathrm{SE}(n)$, representing the poses of robots. We utilize Taylor expansion on Lie groups to transform the proposed safety condition into an online optimization-based safety-filter for control synthesis, where we also derive the curvature correction terms in the constraint to handle highly nonlinear CBFs. To demonstrate the effectiveness of the safety-filter, we apply it to: (i) a differential-drive wheeled robot (whose state evolves on $\mathrm{SE}(2)$) with an open-loop nominal controller, and (ii) a rigid body (evolving on $\mathrm{SE}(3)$) with a closed-loop nominal controller attempting to maintain its initial orientation while translating forward in the presence of a slit.
\end{itemize}
\subsection{Organization} \label{Section:Intro_O}
The remainder of this paper is organized as follows. In Section~\ref{Section:Preliminaries}, we review a few tools from probability theory and Lie group theory that will be used in the paper. In Section~\ref{Section:StochasticCBF}, we propose the state estimation-aware stochastic control barrier function. Then, we provide the corresponding theoretical bound on finite-time safety probability under mild assumptions. In Section~\ref{Section:LTI}, we derive the closed-form expressions of the proposed framework for linear systems with affine CBFs, and extend other two methods to the formulations under state estimation uncertainty, which can be formulated as CBF-QPs for control synthesis and motion planning (or trajectory generation), followed by comparisons among the three methods. In Section~\ref{Section:LieGroup}, we transform the proposed stochastic safety condition into an optimization-based control synthesis formulation on the special Euclidean group, and apply it to $\mathrm{SE}(2)$ and $\mathrm{SE}(3)$ as safety-filters, respectively. Finally, Section~\ref{Section:Conclusion} concludes the paper.

\section{Preliminaries} \label{Section:Preliminaries}
In this section, we briefly review some definitions and results from probability theory (see, e.g., \cite{vershynin2018high,grimmett2020probability}) and Lie group theory (see, e.g., \cite{chirikjian2011stochastic,bullo2019geometric}) that will be used throughout the paper to the extent necessary for the development of the main results. 

\subsection{Probability}
Let $\{\cdot\}$ denote the set of events such that the statement inside the braces holds, and let the abbreviation \emph{a.s.} denote \emph{almost surely}, i.e., an event holds with probability one.
\begin{definition}
[Sub-Gaussian] 
A random variable $X \in \R$ is said to be sub-Gaussian with variance proxy $\bar{\sigma}^2$ if $\forall \varsigma \in \R$,
\begin{equation*}
\E{\Exp{\varsigma (X - \E{X})}}
\leq
\Exp{\frac{1}{2}\varsigma^2\bar{\sigma}^2}.
\end{equation*}
\label{def:sub-Gaussian}
\end{definition}
\noindent The class of sub-Gaussian random variables covers a wide range of light-tailed distributions and will be adopted as a mild assumption on certain terms encoding new information beyond expectations.

\begin{lemma}
[Markov's Inequality] If $X$ is a nonnegative random variable, then for any $\lambda >0$,
\begin{equation*}
\Prob{X \geq \lambda} \leq \frac{\E{X}}{\lambda}.
\end{equation*}
\label{lemma:Markov}
\end{lemma}

\begin{lemma}
[Law of Total Expectation]
Let $(\mathfrak{S},\mathcal{F},\mathbb{P})$ be a probability space where two $\sigma$-algebras $\mathcal{F}_1 \subseteq \mathcal{F}_2$ are defined. For any integrable random variable $X$ (i.e., $\E{|X|} < \infty$),
\begin{align*}
\EC{\EC{X}{\mathcal{F}_2}}{\mathcal{F}_1} = \EC{X}{\mathcal{F}_1} \quad a.s.
\end{align*}
\label{lemma:tower}
\end{lemma}

Inspired by the use of martingale tools in \cite{prajna2004stochastic,steinhardt2012finite,santoyo2019verification,cosner2023robust}, we recall the definitions and a concentration inequality needed for the subsequent analysis as follows.

\begin{definition}
[Martingale]
Let $(\mathfrak{S},\mathcal{F},\mathbb{P})$ be a probability space and $(\mathcal{F}_k)_{k\ge 0}$ be a filtration, i.e., an increasing sequence of sub-$\sigma$-algebras of $\mathcal{F}$. A stochastic process $(X_k)_{k\geq0}$ that is adapted to the filtration (i.e., $X_k$ is $\mathcal{F}_k$-measurable, $\forall k \in \mathbb{N}$) and is integrable at each $k$ is called a martingale if
\begin{equation*}
\EC{X_{k+1}}{\mathcal{F}_k} = X_k \quad \text{a.s.},\quad \forall k \in \mathbb{N},
\end{equation*}
a supermartingale if
$\EC{X_{k+1}}{\mathcal{F}_k} \leq X_k$ a.s., $\forall k \in \mathbb{N}$, or a submartingale if
$\EC{X_{k+1}}{\mathcal{F}_k} \geq X_k$ a.s., $\forall k \in \mathbb{N}$.
\label{def:Martingale}
\end{definition}

\begin{lemma}[Doob's Martingale Inequality]
If the stochastic process $(X_k)_{k\ge 0}$ is a nonnegative supermartingale, then
\begin{equation}
\Prob{\max_{0\leq k\le T} X_k \ge \lambda}
\leq
\frac{\E{X_0}}{\lambda},
\label{eqn:DoobMartingale}
\end{equation}
$\forall \lambda>0$, $\forall\, T \in \mathbb{N}$.
\label{lemma:Doob}
\end{lemma}

\noindent Intuitively, Doob's martingale inequality can be viewed as an extension of Markov's inequality for nonnegative supermartingales, which bounds the probability that a stochastic process ever exceeds a certain level in terms of its initial expectation.

\subsection{Lie Group} \label{Section:LieGrouptool}
Let $\mathcal{G}$ be an $n$-dimensional Lie group (which is also a smooth manifold) with binary operation $\circ$ and identity $\mathcal{I}$, and let $\mathfrak{g} = T_{\mathcal{I}} \mathcal{G}$ denote its Lie algebra (i.e., tangent space at the identity $\mathcal{I}$). For a smooth map $F: \mathcal{G} \to \mathcal{G}$, its pushforward $(dF)_g: T_g\mathcal{G} \to T_{F(g)}\mathcal{G}$ at $g \in \mathcal{G}$ is defined as
\begin{equation*}
(dF)_g (\dot{\gamma}(0)) = \frac{d}{dt}\Big|_{t=0}F(\gamma(t)),
\end{equation*}
for any smooth curve $\gamma(t) \in \mathcal{G}$ such that $\gamma(0) = g$. Each $\xi \in \mathfrak{g}$ induces a left-invariant\footnote{We adopt the left-invariant convention throughout this paper, unless otherwise stated.} vector field
\begin{equation*}
\Gamma_\xi(g) 
=
(d L_g)_{\mathcal{I}} (\xi),
\end{equation*}
where the left translation map $L_g: \mathcal{G}\to \mathcal{G}$ is such that $L_g(G) = g \circ G$, $\forall G \in \mathcal{G}$.

\begin{definition}
[Flow]
For any $\xi \in \mathfrak{g}$, the flow $\Phi^{\Gamma_\xi}_t(g): \mathcal{G}\to \mathcal{G}$ generated by a vector field $\Gamma_\xi$ is such that
\begin{equation*}
\frac{\partial}{\partial t} \Phi^{\Gamma_\xi}_t(g) = \Gamma_\xi (\Phi^{\Gamma_\xi}_t(g)), \quad \Phi^{\Gamma_\xi}_0(g) = g.
\end{equation*}
For notational simplicity, we denote $\Phi_t^\xi \coloneqq \Phi_t^{\Gamma_\xi}$.
\end{definition}

\begin{definition}
[Exp and Log]
The exponential map $\mathrm{Exp}: \mathfrak{g}\to \mathcal{G}$ is defined as
\begin{equation*}
\mathrm{Exp}(\xi) = 
\Phi^\xi_1({\mathcal{I}}),
\end{equation*}
and the logarithm map $\mathrm{Log}: \mathcal{G}\to\mathfrak{g}$ is the 
inverse of $\mathrm{Exp}$.
\end{definition}

For any $\xi \in \mathfrak{g}$ and $g \in \mathcal{G}$, the flow of the left-invariant vector field $\Gamma_\xi$ satisfies 
\begin{equation*}
\Phi_t^\xi(g) = g \circ \mathrm{Exp}(t\xi).
\end{equation*}

The rate of change of a function along the flow generated by a vector field is quantified by the Lie derivative.

\begin{definition}
[Lie derivative]
The Lie derivative of a continuously differentiable function $f: \mathcal{G}\to \R$ along a vector field $\Gamma_\xi$ is defined as
\begin{equation*}
(\mathcal{L}_{\Gamma_\xi} f) (g)
=
\frac{d}{dt}\Big|_{t=0} f(\Phi^\xi_t(g)).
\end{equation*}
For notational simplicity, we denote $\mathcal{L}_{\xi} f \coloneqq \mathcal{L}_{\Gamma_\xi} f$.
\end{definition}

Infinitesimal motions (identified with Lie algebra elements) are transformed under changes of reference by the group adjoint action, and their lack of commutativity is encoded by the Lie bracket.

\begin{definition}
[Adjoint and Lie bracket]
The conjugation map $C_g:\mathcal{G}\to \mathcal{G}$ is defined as
\begin{equation*}
C_g(G) =
g \circ G \circ g^{-1}, \quad \forall G \in \mathcal{G}.
\end{equation*}
The group adjoint action is defined as the pushforward of conjugation at the identity, i.e.,
\begin{equation*}
\mathrm{Ad}_g
=
(dC_g)_{\mathcal{I}} : \mathfrak{g}\to \mathfrak{g},
\end{equation*}
and the Lie bracket is defined as the infinitesimal generator of the group adjoint action,
\begin{equation*}
[\xi,\zeta]
=
\frac{d}{dt}\Big|_{t=0} \mathrm{Ad}_{\mathrm{Exp}(t\xi)} (\zeta), \quad \forall \xi, \zeta \in \mathfrak{g}.
\end{equation*}
\label{def:adjoint}
\end{definition}

How the noncommutativity of infinitesimal motions accumulates into finite transformations is captured by the Baker–Campbell–Hausdorff (BCH) formula, which is essential for the analysis of uncertainty propagation and state estimation on Lie groups.

\begin{lemma}
[BCH formula]
For sufficiently small $\xi, \zeta \in \mathfrak{g}$, the BCH formula gives
\begin{align*}
\mathrm{Log}
\left(
\mathrm{Exp}
(\xi) \circ \mathrm{Exp}
(\zeta)
\right)
=
\xi 
&+ 
\zeta
+ 
\frac{1}{2}[\xi,\zeta] 
+ 
\frac{1}{12}[\xi,[\xi,\zeta] ] \\
&+ 
\frac{1}{12}[\zeta,[\zeta,\xi] ]
+ \mathcal{O}(\|(\xi,\zeta)\|^4).
\end{align*}
\label{lemma:BCH}
\end{lemma}

\begin{definition}
Given an ordered basis $(E_1, \dots, E_n)$ of $\mathfrak{g}$, the isomorphism $(\cdot)^{\vee}: \mathfrak{g} \to \R^n$ is defined as
\begin{equation*}
\left(\sum_{i=1}^n \xi_i E_i\!\right)^{\!\!\vee} 
\!= (\xi_1,\dots,\xi_n)^\top,
\end{equation*}
and the isomorphism $(\cdot)^{\wedge}: \R^n \to \mathfrak{g}$ is the inverse of $(\cdot)^{\vee}$.
\end{definition}

The variation of a smooth function along the flow generated by a vector field is described by the Taylor expansion, which will be leveraged for stochastic safety-critical control synthesis on Lie groups, as will be seen in Section~\ref{Section:LieGroup}.

\begin{lemma}
[Taylor expansion]
For any $g \in \mathcal{G}$, the left-invariant Taylor expansion of
a smooth function $f:\mathcal{G}\to \R$ about $g$ is given by
\begin{align*}
f(g \circ \mathrm{Exp}(\xi))
&=
f(g) 
+
\sum_{i=1}^n  (\mathcal{L}_{E_i}f)
(g) \,\xi_i \\
&\quad\;+
\frac{1}{2}
\sum_{i=1}^n 
\sum_{j=1}^n
(\mathcal{L}_{E_i} \mathcal{L}_{E_j} f)
(g)
\,\xi_i \xi_j
+ \mathcal{O}(\|\xi\|^3).
\end{align*}
\label{lemma:Taylor}
\end{lemma}

The mean and covariance are defined as follows to characterize the statistical moments on Lie groups.
\begin{definition}
[Mean and covariance]
Let $G \in \mathcal{G}$ be a random variable on a Lie group with its probability density function $f_G: \mathcal{G} \to \R_{\geq 0}$, then the mean $\mu = \E{G} \in \mathcal{G}$ is such that
\begin{equation*}
\int_{\mathcal{G}}
\mathrm{Log}^\vee (\mu^{-1} \circ g) \,f_G(g) \, \mathrm{d}g = 0,
\end{equation*}
and the covariance about the mean is defined as
\begin{equation*}
\Sigma =
\int_{\mathcal{G}}
\mathrm{Log}^\vee (\mu^{-1} \circ g)
(\mathrm{Log}^\vee (\mu^{-1} \circ g))^\top
f_G(g) \, \mathrm{d}g,
\end{equation*}
where $\mathrm{d}g$ denotes the Haar measure.
\label{def:geometricmoment}
\end{definition}

\section{Stochastic CBF under State Estimation Uncertainty} \label{Section:StochasticCBF}
Consider a discrete-time stochastic system evolving on an $n$-dimensional smooth manifold $\mathcal{M}$, 
\begin{equation}
\begin{cases}
x_{k+1} = F\big(x_k,u_k,\varepsilon_k\big), \\
z_k = H(x_k,\epsilon_k),
\end{cases}
\label{eqn:system}
\end{equation}
where $x_k \in \mathcal{M}$ denotes the state at time $k$, $u_k \in \mathcal{U} \subset \mathbb{R}^m$ is the control input, the measurement $z_k \in \mathcal{Z}$ takes values on a measurement manifold of certain dimension, $F: \mathcal{M} \times \R^m \times \R^\ell \to \mathcal{M}$ denotes the system dynamics, and $H: \mathcal{M} \times \R^l \to \mathcal{Z}$ denotes the observation model. The process noise $\varepsilon_k \in \R^\ell$ and the measurement noise $\epsilon_k \in \R^l$ are zero-mean Gaussian random variables with covariances $\Sigma_{\varepsilon_k} \in \mathbb{S}^\ell_+$ and $\Sigma_{\epsilon_k} \in \mathbb{S}^l_+$, respectively. They are assumed to be mutually independent and independent across time. In general, the process noise and the measurement noise are not additive with respect to the state and measurement manifolds, respectively.

The observation filtration is defined as $(\mathcal{F}_k)_{k\ge 0}$, where
\begin{equation}
\mathcal{F}_k^- = \sigma\big(x_0,u_{0:k-1},z_{0:k-1}\big)
\subset
\mathcal{F}_k = \sigma\big(x_0,u_{0:k-1},z_{0:k}\big), \notag
\end{equation}
and the state filtration is defined as $(\mathscr{F}_k)_{k\ge 0}$, where
\begin{equation*}
\mathscr{F}_k 
= \sigma\big(x_0,\varepsilon_{0:k-1},\epsilon_{0:k},u_{0:k-1}\big).
\end{equation*}
At each time $k$, the state $x_k$ is $\mathscr{F}_k$-measurable while the control input $u_k$ is designed to be $\mathcal{F}_k$-measurable, $\varepsilon_k$ is independent of $\mathcal{F}_k$, and $\epsilon_k$ is independent of $\mathcal{F}_k^-$, so $\mathcal{F}_{k-1} = \mathcal{F}_k^- \subset \mathcal{F}_k$.

Let $h:\mathcal{M} \to \mathbb{R}$ be of class $C^2$, and the safe set be
\begin{equation}
\mathcal{C} 
= 
\left\{ x\in\mathcal{M} 
\,\big|\,
h(x)\ge 0
\right\}.
\label{eqn:safeset}
\end{equation}
Denote $Y_k = h(x_k)$, which is $\mathscr{F}_k$-measurable. Then, $x_k \in \mathcal{C}$ is equivalent to $Y_k \geq 0$, $\forall k \in \{0,\dots,T\}$ with $T \in \mathbb{N}$. Since the process noise $\varepsilon_k$ and the measurement noise $\epsilon_k$ have unbounded supports, infinite-time safety (i.e., forward invariance) is impossible to achieve \cite{steinhardt2012finite,so2023almost}, so we study the finite-time safety (e.g., \cite{cosner2023robust,santoyo2021barrier}) defined as follows.

\begin{definition}
For any $T \in \mathbb{N}$ and initial state $x_0 \in \mathcal{C}$, 
the $T$-step exit probability is defined as
\begin{equation}
P_{\text{\rm exit}}(T,x_0)
= 
\mathbb{P}
\left\{
\min_{0\leq k\leq T} Y_k < 0
\right\}.
\label{eqn:exit-prob}
\end{equation}
\label{def:exit-prob}
\end{definition}

However, unlike \cite{santoyo2021barrier,cosner2023robust,cosner2024bounding}, where the ground-truth state $x_k$ is assumed to be fully known at every time step $k$, we need to analyze $P_{\text{\rm exit}}(T,x_0)$ while considering the uncertainty propagation of the state estimates. To this end, first, we define the state estimation-aware stochastic CBF (SEA-SCBF) in Definition~\ref{def:seaCBF}.

\begin{definition}[SEA-SCBF]
\label{def:seaCBF}
A function $h \in C^2(\mathcal{M})$ is a state estimation-aware stochastic CBF (SEA-SCBF) with respect to \eqref{eqn:system} and \eqref{eqn:safeset}, if there exists an $\alpha\in(0,1]$ such that at each time step $k$, whenever $\widetilde{Y}_k \geq 0$, there exists an $\mathcal{F}_k$-measurable coefficient $\beta_k\ge 0$ and an $\mathcal{F}_k$-measurable control input $u_k \in \mathcal{U}$ such that
\begin{equation}
\EC{h(x_{k+1})}{\mathcal{F}_k} 
- 
\beta_k
\sqrt{\VarC{h(x_{k+1})}{\mathcal{F}_k}}
\geq \alpha  \widetilde{Y}_k
\quad \text{a.s.},
\label{eqn:SEA-SCBF}
\end{equation}
where $\widetilde{Y}_k = \EC{h(x_{k})}{\mathcal{F}_k}$.
\end{definition}

In general, $\EC{h(x_{k+1})}{\mathcal{F}_k}$ and $\VarC{h(x_{k+1})}{\mathcal{F}_k}$ depend on the control input $u_k$ while $\beta_k$ and $\widetilde{Y}_k$ do not. The constraint \eqref{eqn:SEA-SCBF} is motivated by the intuition that given a safety margin characterized by the posterior estimate of the current barrier function value, actions that lead to a high predicted uncertainty must be compensated by a sufficiently large increase in the predicted barrier function value.

Then, we define two innovation-like random variables
\begin{equation}
\Delta_{k+1}
=
\widetilde{Y}_{k+1}
- \mathbb{E}\big[\, \widetilde{Y}_{k+1} \,\big|\, \mathcal{F}_k\,\big],
\label{eqn:Delta}
\end{equation}
\begin{equation}
\delta_{k} = Y_{k} - \widetilde{Y}_{k}, \label{eqn:delta}
\end{equation}
which are assumed to be sub-Gaussian, conditionally on $\mathcal{F}_k$,  with variance proxies $\bar{\sigma}_{k+1}^2$ and $\bar{\tau}_k^2$, respectively. Based on Lemma~\ref{lemma:tower}, we have
\begin{align*}
\EC{\Delta_{k+1}}{\mathcal{F}_k}
&=
\mathbb{E}\big[\, \widetilde{Y}_{k+1} \,\big|\, \mathcal{F}_k\,\big]
-
\mathbb{E}\big[\, \mathbb{E}\big[\, \widetilde{Y}_{k+1} \,\big|\, \mathcal{F}_k\,\big] \,\big|\, \mathcal{F}_k\,\big] \\
&=0.
\end{align*}
and \eqref{eqn:Delta} becomes
\begin{align}
\Delta_{k+1}
=\,&
\widetilde{Y}_{k+1}
- 
\EC{\EC{Y_{k+1}}{\mathcal{F}_{k+1}}}{\mathcal{F}_k}
\notag \\
=\,&
\widetilde{Y}_{k+1} 
- 
\EC{Y_{k+1}}{\mathcal{F}_k}. \notag
\end{align}

In addition, we define the stochastic process $(M_k)_{k\geq0}$, where
\begin{equation}
M_k =
\begin{cases}
0, & k=0, \\
\sum_{i=0}^{k-1}\alpha^{-(i+1)}\Delta_{i+1},& k\geq 1. 
\end{cases}
\label{eqn:M_k}
\end{equation}
According to \eqref{eqn:M_k}, we have
\begin{equation}
M_{k+1} = \sum_{i=0}^{k}\alpha^{-(i+1)}\Delta_{i+1} = M_k + \alpha^{-(k+1)}\Delta_{k+1}.
\label{eqn:M_k_2}
\end{equation}
Taking expectation conditioned on $\mathcal{F}_k$ on both sides of \eqref{eqn:M_k_2} yields
\begin{equation*}
\EC{M_{k+1}}{\mathcal{F}_k}
=
\EC{M_k}{\mathcal{F}_k}
+
\alpha^{-(k+1)}\,
\EC{\Delta_{k+1}}{\mathcal{F}_k}.
\end{equation*}
Since $M_k$ and $\mathbb{E}\big[\, \widetilde{Y}_{k+1} \,\big|\, \mathcal{F}_k\,\big]$ are $\mathcal{F}_k$-measurable, and $\EC{\Delta_{k+1}}{\mathcal{F}_k} = 0$, then we have $\EC{M_k}{\mathcal{F}_k} = M_k$, and $\EC{M_{k+1}}{\mathcal{F}_k}
= M_k$ a.s., $\forall k \in \mathbb{N}$. Hence, $(M_k)_{k\geq0}$ is a martingale. (Note that $\EC{M_1}{\mathcal{F}_0} = M_0 = 0$, whereas $\EC{M_{k+1}}{\mathcal{F}_k} \neq 0$, $\forall k \geq 1$.) 

Next, we construct a nonnegative supermartingale related to the SEA-SCBF value $Y_k$ as follows, so that the concentration inequality \eqref{eqn:DoobMartingale} can be applied afterwards.
\begin{lemma}
Given $\tau \in \R$, the stochastic process $(S_k)_{k\geq0}$, defined as
\begin{equation}
S_k =
\Exp{-\tau M_k- \frac{1}{2}\tau^2 V_k},
\label{eqn:S_k}
\end{equation}
in which
\begin{equation}
V_k =
\begin{cases}
0, & k=0, \\
\sum_{i=0}^{k-1}
\bar{\sigma}_{i+1}^2\alpha^{-2(i+1)},& k\geq 1, 
\end{cases}
\label{eqn:V_k}
\end{equation}
is a nonnegative supermartingale.
\label{lemma:supermartingale}
\end{lemma}
\begin{proof}
Apparently, $(S_k)_{k\geq0}$ is nonnegative. Since $\Delta_k$ is $\mathcal{F}_k$-measurable, then $S_k$ is $\mathcal{F}_k$-measurable, $\forall k \in \mathbb{N}$. Using \eqref{eqn:M_k_2}, we can obtain
\begin{align}
S_{k+1}
&= 
\Exp{-\tau M_{k+1}-\frac{1}{2}\tau^2
V_{k+1}} \notag \\
&=
S_k
\cdot \,
\Exp{
-\tau \alpha^{-(k+1)}\Delta_{k+1} 
-\frac{1}{2}\tau^2
\bar{\sigma}_{k+1}^2\alpha^{-2(k+1)}}.
\label{eqn:Lk_derive}
\end{align}
Based on the $\mathcal{F}_k$-measurability of $S_k$, taking expectation conditioned on $\mathcal{F}_k$ on both sides of \eqref{eqn:Lk_derive} yields
\begin{align}
\EC{S_{k+1}}{\mathcal{F}_k}
=\,
&S_k 
\cdot
\Exp{- \frac{1}{2}\tau^2\bar{\sigma}_{k+1}^2\alpha^{-2(k+1)}} \notag \\
&\quad
\cdot
\EC{\Exp{-\tau\alpha^{-(k+1)}\Delta_{k+1}}}{\mathcal{F}_k}.
\label{eqn:E_Lk+1}
\end{align}
Based on the assumption that $\Delta_{k+1}$ follows sub-Gaussian distribution conditionally on $\mathcal{F}_k$, there exists a deterministic sequence $(\bar{\sigma}_{k+1}^2)_{k\ge 0}$ such that for any $\varsigma \in \mathbb{R}$,
\begin{equation}
\EC{\Exp{\varsigma\Delta_{k+1}}}{\mathcal{F}_k}
\leq
\Exp{\frac{1}{2}\varsigma^2\bar{\sigma}_{k+1}^2},\notag
\end{equation}
and thus we have
\begin{align}
&\EC{\Exp{-\tau\alpha^{-(k+1)}\Delta_{k+1}}}{\mathcal{F}_k} \notag \\
&\leq
\Exp{\frac{1}{2}\tau^2\alpha^{-2(k+1)}\bar{\sigma}_{k+1}^2}.
\label{eqn:subG-increment}
\end{align}
Combining \eqref{eqn:E_Lk+1} and \eqref{eqn:subG-increment} results in $\EC{S_{k+1}}{\mathcal{F}_k} \leq S_k$ a.s., $\forall k \in \mathbb{N}$. Hence, $(S_k)_{k\geq0}$ is a nonnegative supermartingale.
\end{proof}

We now construct a sub-Gaussian-like bound on $M_k$.
\begin{lemma}
For any $\lambda>0$ and any
$T\in\mathbb{N}$,
\begin{equation}
\Prob{
\max_{0\leq k \leq T}(-M_k)\geq \lambda
}
\leq
\Exp{-\frac{\lambda^2}{2V_T}},
\label{eqn:Sk_bound}
\end{equation}
where $V_T$ is given by \eqref{eqn:V_k}.
\label{lemma:Sk_bound}
\end{lemma}
\begin{proof}
Given $\nu>0$, $\lambda>0$, we define
\begin{equation*}
\psi(\nu,\lambda)
= 
\Exp{\nu\lambda - \frac{1}{2}\nu^2 V_T}.
\end{equation*}
If $\max_{0\le k\le T}(-M_k) \geq \lambda$, then $\exists k^* \leq T$ such
that $-M_{k^*}\ge\lambda$ and $V_{k^*}\le V_T$. Thus,
\begin{align}
S_{k^*}
&=
\Exp{-\nu M_{k^*}- \frac{1}{2}\nu^2 V_{k^*}}
\notag \\
&\geq
\Exp{\nu\lambda - \frac{1}{2}\nu^2 V_T}
= \psi(\nu,\lambda). \notag
\end{align}
Hence, we have
\begin{equation*}
\left\{
\max_{0\le k\leq T}(-M_k) \geq \lambda
\right\}
\subseteq
\left\{
\max_{0\leq k\leq T}S_k \geq \psi(\nu,\lambda)
\right\}.
\end{equation*}
According to Lemma~\ref{lemma:supermartingale}, the process $(S_k)_{k\geq0}$ given by \eqref{eqn:S_k} is a nonnegative supermartingale, and thus applying Lemma~\ref{lemma:Doob} yields
\begin{align}
\Prob{\max_{0\le k\le T}(-M_k) \ge \lambda} \notag
&\leq
\Prob{\max_{0\le k\le T}S_k \ge \psi(\nu,\lambda) }
\notag \\
&\leq
\frac{\E{S_0}}{\psi(\nu,\lambda)}
\notag \\
&= 
\Exp{-\nu\lambda + \frac{1}{2}\nu^2 V_T},
\label{eqn:Sk_bound_derive}
\end{align}
which holds for any $\nu>0$. Without loss of generality, taking
$\nu = \lambda/V_T$ gives $-\nu\lambda 
+ 
\frac{1}{2}\nu^2 V_T
= 
-\frac{\lambda^2}{V_T}
+ 
\frac{1}{2}\frac{\lambda^2}{V_T}
= 
-\frac{\lambda^2}{2V_T}$. Then, \eqref{eqn:Sk_bound_derive} becomes \eqref{eqn:Sk_bound}.
\end{proof}

Finally, we propose the theoretical upper bound on $P_{\text{\rm exit}}(T,x_0)$ as follows.
\begin{theorem}
If $h$ is an SEA-SCBF per {\rm Definition~\ref{def:seaCBF}}, \eqref{eqn:SEA-SCBF} holds for all $k \leq T$ whenever $\widetilde{Y}_k \geq 0$, and $\Delta_{k+1}$ and $\delta_k$ are sub-Gaussian random variables, conditionally on $\mathcal{F}_k$, with proxies $\bar{\sigma}_{k+1}^2$ and $\bar{\tau}_k^2$, respectively, then for any $T \in \mathbb{N}$ and any $0 < \eta < \alpha^T \widetilde{Y}_0$,
\begin{equation}
P_{\text{\rm exit}}(T,x_0) 
\leq 
\Exp{-\frac{(\widetilde{Y}_0 - \alpha^{-T}\eta)^2}{2V_T}} 
+ \sum_{k=0}^T \Exp{-\frac{\eta^2}{2\bar{\tau}_k^2}}.
\label{eqn:P_bound}
\end{equation}
where $V_T 
= \sum_{i=0}^{T-1} \bar{\sigma}_{i+1}^2\,\alpha^{-2(i+1)}$.
\label{theorem:certificate}
\end{theorem}
\begin{proof}
We relate the events such that $Y_k<0$ to the events
related to $\widetilde{Y}_k$ and  $\delta_k$. Given $\eta>0$, for each $k \geq 0$, we have
\begin{equation}
\left\{Y_k<0\right\}
= 
\left\{\widetilde{Y}_k 
+
\delta_k < 0
\right\}
\subseteq
\left\{\widetilde{Y}_k\le\eta
\right\}
\cup 
\left\{
\delta_k\le -\eta
\right\}.
\notag
\end{equation}
Then, we can obtain
\begin{equation*}
\displaystyle \bigcup_{k=0}^T\{Y_k<0\}
\subseteq
\left(
\displaystyle \bigcup_{k=0}^T
\left\{
\widetilde{Y}_k\leq\eta
\right\}
\right)
\displaystyle \bigcup
\left(
\displaystyle \bigcup_{k=0}^T\{\delta_k\leq-\eta\}
\right).
\end{equation*}
Hence, we have
\begin{equation}
P_{\text{\rm exit}}(T,x_0)
\leq
\Prob{\min_{0\le k\le T}\widetilde{Y}_k\le\eta}
+
\Prob{\min_{0\le k\le T}\delta_k\le-\eta}.
\label{eqn:exitsplit}
\end{equation}
Now we can bound the two terms on the right hand side of \eqref{eqn:exitsplit} separately. 

To begin with, we define $\widetilde{S}_k = \alpha^{-k}\,\widetilde{Y}_k$, and denote $\mu_{h_{k+1}}^- 
= 
\EC{Y_{k+1}}{\mathcal{F}_k}$ and $(\sigma_{h_{k+1}}^-)^2 
= 
\VarC{h(x_{k+1})}{\mathcal{F}_k}$. Based on the decomposition $\widetilde{Y}_{k+1} = \mu_{h_{k+1}}^- + \Delta_{k+1}$ and \eqref{eqn:M_k_2}, we can obtain
\begin{align}
\widetilde{S}_{k+1}
&= 
\alpha^{-(k+1)}\,
\mu_{h_{k+1}}^-
+ 
\alpha^{-(k+1)}
\Delta_{k+1} \notag \\ 
&= 
\alpha^{-(k+1)}
\mu_{h_{k+1}}^-
+ 
(M_{k+1}-M_k).
\label{eqn:eqn:M_k_3}
\end{align}
Combining \eqref{eqn:eqn:M_k_3} with \eqref{eqn:SEA-SCBF} yields
\begin{equation*}
\widetilde{S}_{k+1}
\geq
\widetilde{S}_k
+ 
\alpha^{-(k+1)}
\beta_k
\sigma_{h_{k+1}}^-
+ 
(M_{k+1}-M_k), 
\end{equation*}
and thus we can obtain
\begin{equation}
\widetilde{S}_{k+1} - M_{k+1}
\geq
\widetilde{S}_k - M_k
+ 
\alpha^{-(k+1)}
\beta_k
\sigma_{h_{k+1}}^-.
\label{eqn:M_k_4}
\end{equation}
Based on \eqref{eqn:M_k_4}, we can get
\begin{equation}
\widetilde{S}_k
\geq
\widetilde{S}_0
+ M_k
+ \sum_{i=0}^{k-1}\alpha^{-(i+1)}
\beta_k
\sigma_{h_{k+1}}^-.
\notag
\end{equation}
Since $\sum_{i=0}^{k-1}\alpha^{-(i+1)}
\beta_k
\sigma_{h_{k+1}}^- \geq 0$, then
\begin{equation}
\widetilde{S}_k 
\geq
\widetilde{S}_0 + M_k,
\quad \forall k\ge 0.
\label{eqn:M_k_simple}
\end{equation}
Suppose for some $\bar{k} \leq T$, $\widetilde{Y}_{\bar{k}}\le\eta$, then we have
\begin{equation}
\widetilde{S}_{\bar{k}} = \alpha^{-\bar{k}}\widetilde{Y}_{\bar{k}}
\leq \alpha^{-\bar{k}}\eta
\leq \alpha^{-T}\eta.
\notag
\end{equation}
Based on \eqref{eqn:M_k_simple}, we have $
\widetilde{S}_0 + M_{\bar{k}} \leq \widetilde{S}_{\bar{k}} \leq \alpha^{-T}\eta$, so
$M_{\bar{k}} \leq \alpha^{-T}\eta - \widetilde{S}_0$. Hence, we can obtain
\begin{align*}
\left\{
\min_{0\le k\leq T}\widetilde{Y}_k\leq\eta\right\}
&\subseteq
\left\{
\min_{0\leq k\leq T}M_k\leq\alpha^{-T}\eta - \widetilde{S}_0
\right\} \notag \\
&=
\left\{
\max_{0\leq k\leq T}(-M_k)\geq \widetilde{S}_0 - \alpha^{-T}\eta
\right\}.
\end{align*}

Applying Lemma~\ref{lemma:Sk_bound} with
$\lambda = \widetilde{S}_0 - \alpha^{-T}\eta >0$ yields
\begin{equation}
\Prob{
\min_{0\le k\le T}\widetilde{Y}_k\le\eta
}
\leq
\Exp{-\frac{(\widetilde{S}_0 - \alpha^{-T}\eta)^2}{2 V_T}},
\label{eqn:bound1}
\end{equation}
where $\widetilde{S}_0 = \alpha^{0}\widetilde{Y}_0 = \widetilde{Y}_0$.

Next, since $\delta_k$ is assumed to follow sub-Gaussian distribution conditionally on $\mathcal{F}_k$, there exists a deterministic sequence $(\bar{\tau}_k^2)_{k\ge 0}$ such that
for any $\varsigma \in \mathbb{R}$,
\begin{equation*}
\EC{\Exp{\varsigma\delta_k}}{\mathcal{F}_k}
\le
\Exp{
\frac{1}{2}\varsigma^2\bar\tau_k^2}.
\end{equation*}
Hence, based on Lemma~\ref{lemma:Markov}, given any $k \geq 0$, $\eta>0$, and $\tau>0$, we have
\begin{align*}
\Prob{\delta_k\le -\eta\,\middle|\,\mathcal{F}_k}
&= 
\Prob{\Exp{-\tau\delta_k}\geq \Exp{\tau\eta}
\,\middle|\,
\mathcal{F}_k} \notag \\
&\leq
\Exp{-\tau\eta}
\EC{\Exp{-\tau\delta_k}}{\mathcal{F}_k} \notag \\
&\leq
\Exp{-\tau\eta + \frac{1}{2}\tau^2\bar\tau_k^2}.
\end{align*}
Taking $\tau = \eta/\bar\tau_k^2$ gives $-\tau\eta + \frac{1}{2}\tau^2\bar\tau_k^2
  = -\frac{\eta^2}{\bar\tau_k^2}
    + \frac12\frac{\eta^2}{\bar\tau_k^2}
  = -\frac{\eta^2}{2\bar\tau_k^2}$. Thus, we can obtain
\begin{equation}
\Prob{\delta_k\le -\eta\,\middle|\,\mathcal{F}_k}
\leq
\Exp{-\frac{\eta^2}{2\bar\tau_k^2}}.
\label{eqn:deltabound1}
\end{equation}
Taking expectation on both sides of \eqref{eqn:deltabound1} and using Lemma~\ref{lemma:tower} yields
\begin{align*}
\Prob{\delta_k \leq - \eta} 
&= \E{\Prob{\delta_k\leq -\eta\,\middle|\,\mathcal{F}_k}} \notag \\
& \leq
\E{\Exp{-\frac{\eta^2}{2\bar{\tau}_k^2}}} 
= 
\Exp{-\frac{\eta^2}{2\bar{\tau}_k^2}}.
\end{align*}
Then, we can obtain
\begin{align}
\Prob{\min_{0\le k\le T}\delta_k\le -\eta}
&= 
\Prob{
\bigcup_{k=0}^T 
\left\{
\delta_k \leq -\eta
\right\}
} \notag \\
&\leq
\sum_{k=0}^T
\Prob{\delta_k\le -\eta}
\notag \\
&\leq
\sum_{k=0}^T \Exp{-\frac{\eta^2}{2\bar{\tau}_k^2}}.
\label{eqn:prob_delta}
\end{align}
Therefore, combining \eqref{eqn:exitsplit}, \eqref{eqn:bound1}, and \eqref{eqn:prob_delta} concludes the proof of Theorem~\ref{theorem:certificate}.
\end{proof}
Theorem~\ref{theorem:certificate} provides an offline-computable bound of finite-time stochastic safety incorporating state estimation uncertainty. Once the proxies $\bar{\sigma}_{k+1}^2$ and $\bar{\tau}_k^2$ are specified, the right-hand side of \eqref{eqn:P_bound} can be computed without access to actual trajectories (i.e., without running experiments).

A central difficulty of bounding the finite-time exit probability \eqref{eqn:exit-prob} is that the true barrier value $h(x_k)$ is adapted to the state filtration $\mathscr{F}_k$, which is unknown, whereas the controller is adapted to the observation filtration $\mathcal{F}_k$. Our analysis thus bounds the unobservable exit event $h(x_k)<0$ by separating it into two risk sources. The first term on the right-hand side of \eqref{eqn:P_bound} quantifies the probability that the posterior estimate of the barrier value drops below a chosen threshold $\eta$. The intuition is that a relatively large initial margin $\widetilde{Y}_0$ exponentially suppresses the risk, while large $V_T$ accumulates the effect of innovations of the posterior barrier estimate $\Delta_{k+1}$ over time horizon, reflecting that long horizons and high predictive uncertainty inevitably loosen finite-time certificates. The second term quantifies the probability that the estimation error $\delta_k$ drives the true barrier value negative even when $\widetilde{Y}_k$ appears safe. The tunable threshold $\eta$ trades off these two risk sources, and it can be obtained by minimizing the scalar objective function in $\eta$, 
\begin{equation}
\eta^* 
= 
\argmin_{0<\eta<\alpha^T\widetilde{Y}_0} f(\eta),
\label{eqn:eta_opt}
\end{equation}
where 
$
f(\eta) = \Exp{-\frac{(\widetilde{Y}_0 - \alpha^{-T}\eta)^2}{2V_T}} 
+ \sum_{k=0}^T \Exp{-\frac{\eta^2}{2\bar{\tau}_k^2}},
$
which can be solved efficiently by a line search.

As such, we propose the safety-critical control framework
\vspace{5pt}
\hrule height 0.8pt
\vspace{7pt}
\noindent\textbf{SEA-SCBF Control Synthesis:}
\begin{equation}
\begin{aligned}
& \argmin_{{u}_k \in \mathcal{U}} 
& & \|u_k - \bar{u}_k\|^2 \\
& \quad\;
\textnormal{s.t.}
& &  
\mu_{h_{k+1}}^-\!(u_k) 
- 
\beta_k \, \sigma_{h_{k+1}}^-\!(u_k)
\geq
\alpha \mu_{h_k},
\end{aligned}
\label{eqn:SE-CBF_Controller}
\end{equation}
in which
\begin{align*}
\mu_{h_{k+1}}^-\!(u_k) 
&= 
\EC{h(x_{k+1})}{\mathcal{F}_k},\\
\sigma_{h_{k+1}}^-\!(u_k)
&= 
\sqrt{\VarC{h(x_{k+1})}{\mathcal{F}_k}},\\
\mu_{h_k} 
&= 
\EC{h(x_{k})}{\mathcal{F}_k}.
\end{align*}
\hrule height 0.8pt
\vspace{7pt}
\noindent The online control synthesis framework \eqref{eqn:SE-CBF_Controller} is adaptive to the level of uncertainty arising from both process and measurement noise, supports real-time computation, and 
enables probabilistic reasoning about finite-time safety.

Note that the probabilistic bound \eqref{eqn:P_bound} is derived under mild structural requirements, i.e., sub-Gaussian assumptions on $\Delta_{k+1}$ and $\delta_{k+1}$. It does not rely on, for example, an upper-bounded barrier function $h$, a polynomial barrier function $h$, bounded noise terms $\varepsilon_k$, $\epsilon_k$, or bounded predictable quadratic variations. At this level of generality, characterizing the tightest possible bound is beyond the scope of this work. If stronger information is assumed, then the same proof pipeline accommodates alternative martingale tools (see, e.g., \cite{freedman1975tail}), potentially yielding tighter bounds. Note that such refinements affect only the certification module, while the proposed control synthesis framework \eqref{eqn:SE-CBF_Controller} remains unchanged. In addition, 
obtaining the proxies $\bar{\sigma}_{k+1}^2$ and $\bar{\tau}_k^2$  falls under system identification or distribution learning, and any model-based or data-driven methods can be plugged in. Therefore, in Section~\ref{Section:LTI_AffineCBF}, we  focus on demonstrating that the proxies can be calculated in closed-form in a special setting, and the resulting probabilistic bound exhibits the correct qualitative behavior without positioning it as a trajectory-level predictor of exact exit frequencies.

\section{SEA-CBFs for Linear Systems} \label{Section:LTI}
In this section, first, we derive the closed-form expressions of \eqref{eqn:SE-CBF_Controller}, $\bar{\sigma}_{k+1}^2$, and $\bar{\tau}_k^2$, for linear system with affine CBF. Then, we extend the methods in \cite{cosner2023robust,cosner2024bounding,mestres2025probabilistic} to the formulations under state estimation uncertainty. Finally, based on the results, we propose a new framework for motion planning (or trajectory generation), which provides a complementary perspective to sampling-based methods (see, e.g., \cite{lavalle1998rapidly,kavraki2002probabilistic}) and diffusion-model-based methods (see, e.g., \cite{janner2022planning}).

\subsection{Linear System with Affine CBF} \label{Section:LTI_AffineCBF}
Consider the linear stochastic system
\begin{equation}
\begin{cases}
p_{k+1} = Ap_k + B u_k + \varepsilon_k, \\
z_k = H p_k + \epsilon_k,
\end{cases}
\label{eqn:LTI}
\end{equation}
where $p_k \in \mathbb{R}^n$, $z_k \in \mathbb{R}^d$, $A\in\R^{n\times n}$, $B\in\R^{n\times m}$, $H\in\R^{d\times n}$, $\varepsilon_k \sim \mathcal{N}(0, \Sigma_{\varepsilon_k})$, $\epsilon_k \sim \mathcal{N}(0, \Sigma_{\epsilon_k})$, 
$\Sigma_{\varepsilon_k} \in \mathbb{S}^n_+$, and $\Sigma_{\epsilon_k} \in \mathbb{S}^d_+$.

At each time step $k$, Kalman filter \cite{kalman1960new} gives $\mu_k = \EC{p_k}{\mathcal{F}_k}$ and $\Sigma_k = \CovC{p_k}{\mathcal{F}_k}$. If the SEA-SCBF is affine in state, i.e., 
\begin{equation}
h(p) = c^\top p - b,
\label{eqn:affineCBF}
\end{equation}
where $c \in \R^n$, $b \in \R$, we can obtain
\begin{align*}
\EC{h(p_{k})}{\mathcal{F}_k} = c^\top\mu_k-b.
\end{align*}
\begin{align*}
 \EC{h(p_{k+1})}{\mathcal{F}_k} 
&= 
\EC{c^\top p_{k+1} - b}{\mathcal{F}_k} \\
&= c^\top (A \mu_k + B u_k) - b,
\end{align*}
\begin{align*}
\VarC{h(p_{k+1})}{\mathcal{F}_k} 
&= \VarC{c^\top p_{k+1}}{\mathcal{F}_k}\\
&= c^\top 
\left(
A \Sigma_k A^\top + \Sigma_{\varepsilon_k}
\right) c.
\end{align*}
Then, the constraint in the SEA-SCBF control synthesis framework \eqref{eqn:SE-CBF_Controller} becomes
\begin{equation*}
c^\top B\, u_k
\geq
m(\mu_k)
+
\rho(\Sigma_k, \Sigma_{\varepsilon_k}),
\end{equation*}
in which
\begin{equation}
m(\mu_k) 
= 
(1-\alpha)\,b 
- 
c^\top(A-\alpha I)\,\mu_k, 
\label{eqn:m(mu)}
\end{equation}
\begin{equation}
\rho(\Sigma_k, \Sigma_{\varepsilon_k})
= 
\beta_k
\sqrt{
c^\top
\left(
A\Sigma_k A^\top + \Sigma_{\varepsilon_k}
\right)
c}.
\label{eqn:rho(sig)}
\end{equation}
As such, at each time step, given a nominal control input $\bar{u}_k \in \R^m$, the stochastic safety-critical controller under state estimation uncertainty can be synthesized via
\begin{equation}
\begin{aligned}
& \argmin_{{u}_k \in \mathcal{U}} 
& & \|u_k - \bar{u}_k\|^2 \\
& \quad\;
\textnormal{s.t.}
& &  
c^\top B \,u_k
\geq
m(\mu_k)
+
\rho(\Sigma_k, \Sigma_{\varepsilon_k}),
\end{aligned}
\label{eqn:SEASCBFQP}
\end{equation}
which is a quadratic program (QP).

Next, we derive the closed-form expressions for $\bar{\sigma}_{k+1}^2$ and $\bar{\tau}_k^2$. Since the posterior estimate of the state $p_k$ is Gaussian with mean $\mu_k$ and covariance $\Sigma_k$, then ${Y}_k$ conditioned on $\mathcal{F}_k$ is also a Gaussian (after linear transformation) with mean $c^\top \mu_k - b$ and covariance $c^\top \Sigma_k c$. Hence, $\delta_k = c^\top (p_k - \mu_k)$, and thus
\begin{equation}
\bar{\tau}_k^2 = c^\top \Sigma_k c.
\label{eqn:proxyforLTI}
\end{equation}
Likewise, since $\widetilde{Y}_{k+1} = c^\top \mu_{k+1} - b$, then 
\begin{equation*}
\Delta_{k+1} = c^\top K_{k+1}\,r_{k+1},
\end{equation*}
in which the Kalman gain
\begin{equation*}
K_{k+1} = \Sigma_{k+1}^- H^\top (H \Sigma_{k+1}^- H^\top + \Sigma_{\epsilon_{k+1}})^{-1},
\end{equation*}
where $\Sigma_{k+1}^- = A \Sigma_k A^\top + \Sigma_{\varepsilon_k}$ 
and $r_{k+1} 
=
H(p_{k+1} - A \mu_k - B u_k) + \varepsilon_{k+1}$. Since $p_{k+1}$ conditionally on $\mathcal{F}_k$ and $\varepsilon_{k+1}$ are independent, then we can obtain
\begin{equation}
\bar{\sigma}_{k+1}^2 = c^\top K_{k+1} (H \Sigma_{k+1}^- H^\top + \Sigma_{\epsilon_{k+1}}) K_{k+1}^\top c.
\label{eqn:Delta_proxy}
\end{equation}

Now that we have the closed-form expressions for $\bar{\sigma}_{k+1}^2$ and $\bar{\tau}_k^2$ which can be calculated offline via Riccati recursion as seen in \eqref{eqn:proxyforLTI} and \eqref{eqn:Delta_proxy}, as discussed at the end of Section~\ref{Section:StochasticCBF} we demonstrate the behavior of the bound in Theorem~\ref{theorem:certificate} by applying the SEA-SCBF-QP \eqref{eqn:SEASCBFQP} to a linear system with affine barrier function. Specifically, we consider \eqref{eqn:LTI} with $p_k \in \R^2$, $A = I$, $B = I\Delta t$, $H = I$, $\epsilon_k \sim \mathcal{N}(0,0.03^2I)$, $c = [0, 1]^\top$, $b = -0.5$, $\alpha = 1$, and $\beta_k = \beta \Exp{-7 \widetilde{Y}_k}$. To clearly illustrate safety violations in the Monte Carlo (MC) experiments, we set the level of some noise to be relatively higher than the system scale, and adopt an open-loop controller as the nominal control input that drives the state along the positive $x$-axis. Thus we let $\varepsilon_k \sim \mathcal{N}(0,\diag(0.03^2, \sigma_y^2))$ to show the influence of the noise level in $y$-axis on the behavior of $P_{\text{\rm exit}}(T,x_0)$. In addition, $\eta$ is obtained through \eqref{eqn:eta_opt}. 

\begin{figure}[t]
\centering
\subfigure[]{
\begin{minipage}[b]{0.22\textwidth}
\includegraphics[width=1\textwidth]{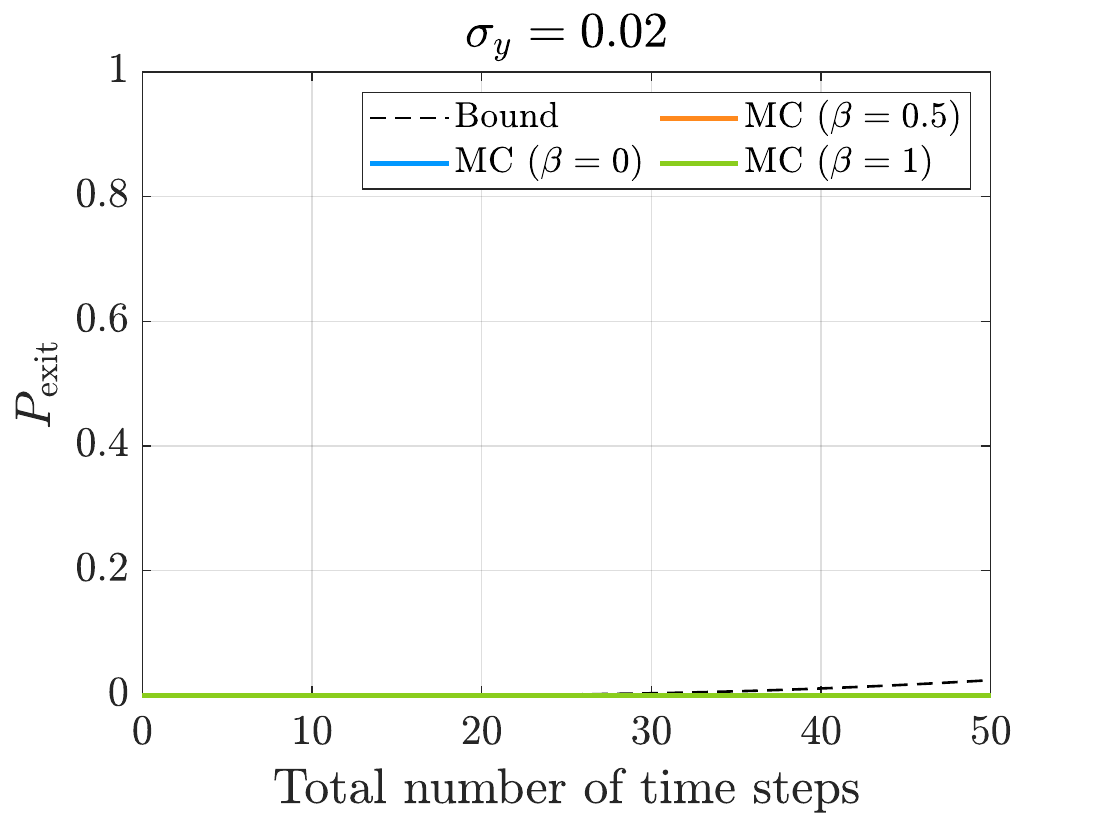}
\end{minipage}
\label{fig:bound_0}
}
\subfigure[]{
\begin{minipage}[b]{0.22\textwidth}
\includegraphics[width=1\textwidth]{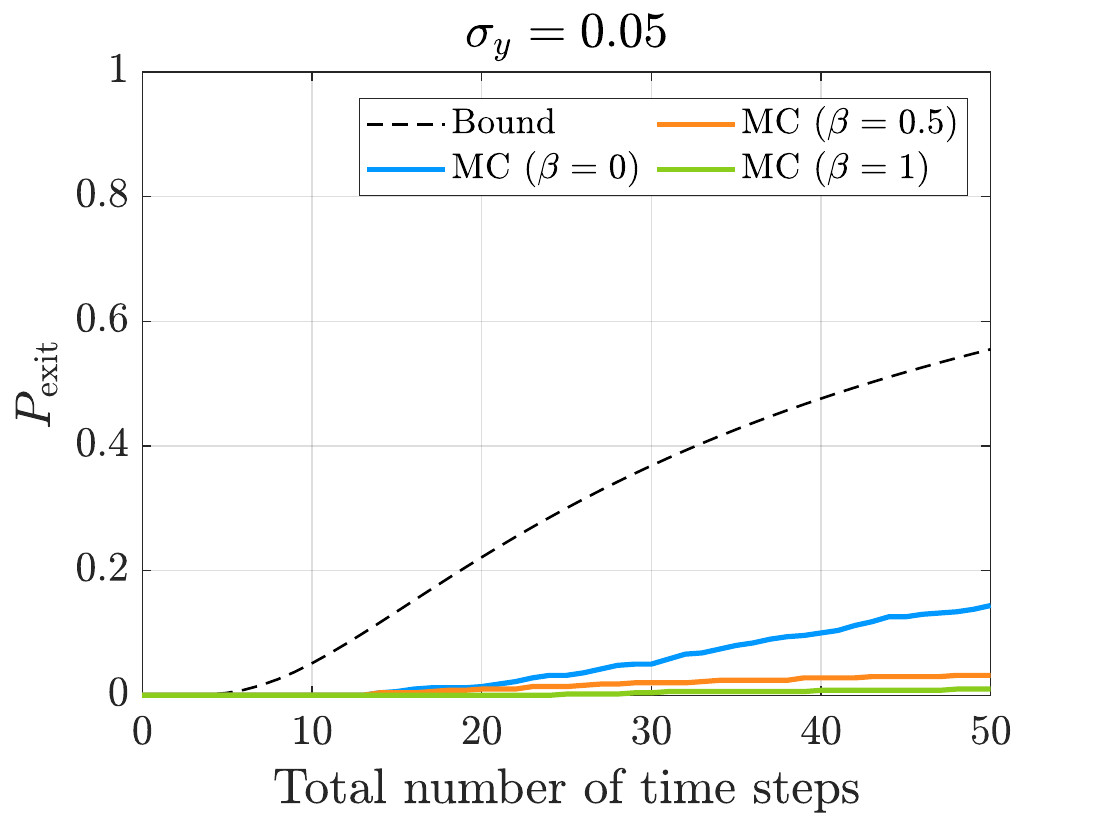}
\end{minipage}
\label{fig:bound_1}
}
\subfigure[]{
\begin{minipage}[b]{0.22\textwidth}
\includegraphics[width=1\textwidth]{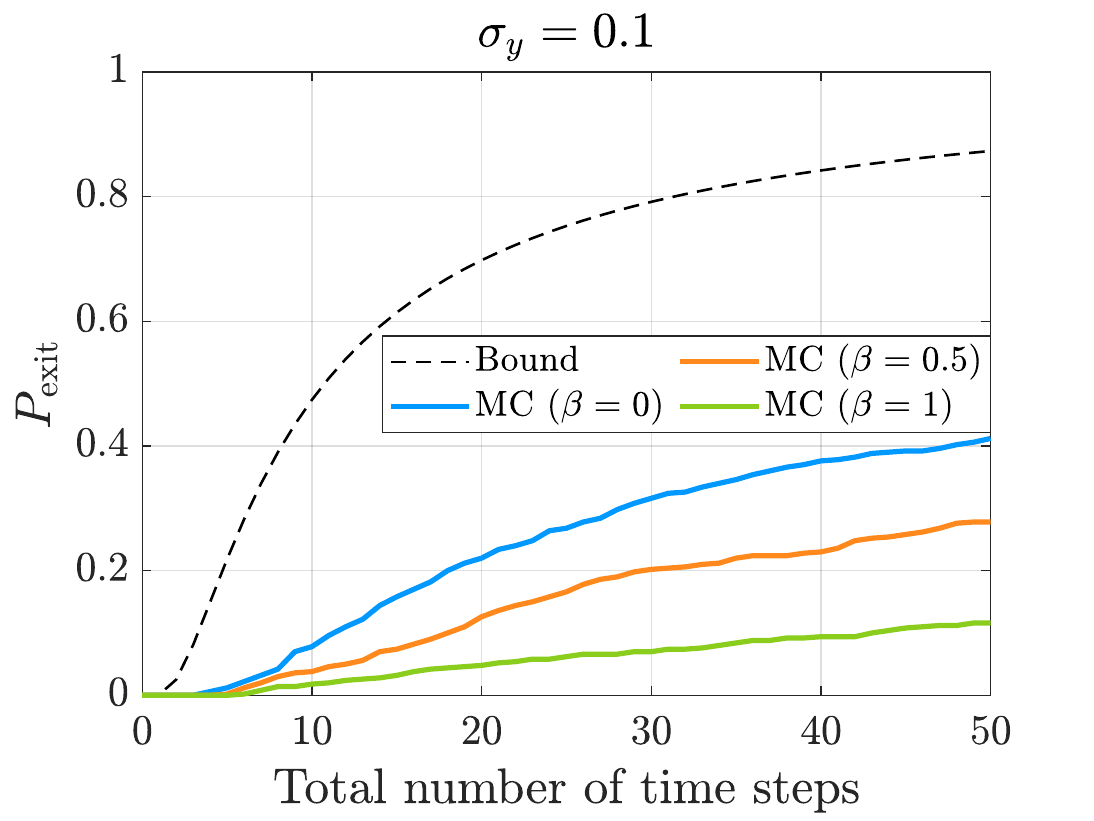}
\end{minipage}
\label{fig:bound_1_1}
}
\subfigure[]{
\begin{minipage}[b]{0.22\textwidth}
\includegraphics[width=1\textwidth]{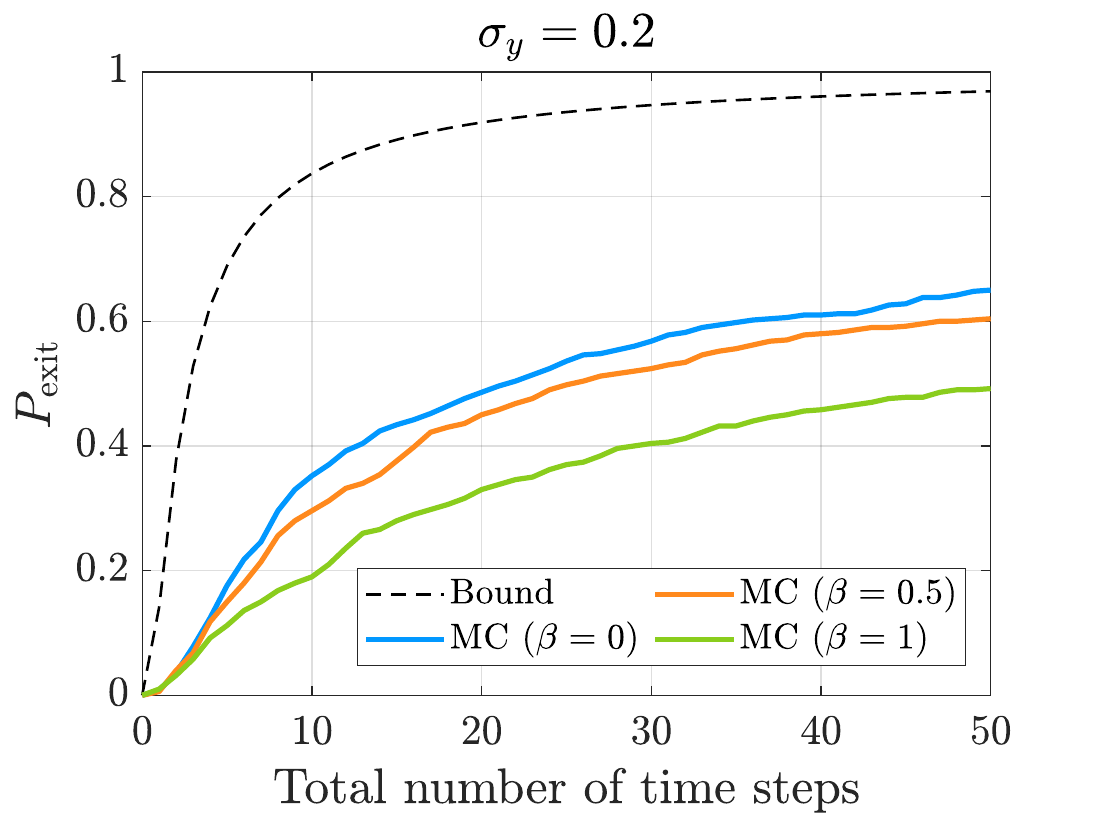}
\end{minipage}
\label{fig:bound_2}
}
\caption{Comparison between the theoretical upper bound on $P_{\text{\rm exit}}(T,x_0)$ per \eqref{eqn:P_bound} and the $T$-step exit frequency per \eqref{eqn:exit-prob} along the total number of steps $T$ under different parameter settings (with 500 MC trials).
}
\label{fig:Pbounds}
\end{figure}

The simulation results (with 500 MC trials) are shown in Fig.~\ref{fig:Pbounds}, from which one can notice that for a fixed total number of time steps $T$, higher noise leads to a higher probability of exit from the safe set. In addition, the theoretical upper bound on $P_{\text{\rm exit}}(T,x_0)$ per \eqref{eqn:P_bound} is always higher than the $T$-step exit frequency per \eqref{eqn:exit-prob}, which verifies Theorem~\ref{theorem:certificate}. Moreover, both of them increase as the total number of time steps $T$ increases. This is because sufficiently large noise is more likely to be sampled over longer time horizons, potentially leading to exit from the safe set. Furthermore, since $\beta_k$ is $\mathcal{F}_k$-measurable, it does not affect the theoretical upper bound on $P_{\text{\rm exit}}(T,x_0)$ whereas it affects the MC experiments. Fig.~\ref{fig:Pbounds} implies that larger $\beta$ results in safer behavior. This is because for linear system with affine barrier function, as seen in \eqref{eqn:rho(sig)}, $\sigma_{h_{k+1}}^-$ does not depend on the control input $u_k$. Thus, given a posterior estimate of the current barrier function value, larger $\beta$ leads to larger $\beta_k$, resulting in larger $\rho$, and therefore the controller $u_k$ becomes more conservative, leading to safer behavior.

\subsection{SCBF-based Motion Planning} \label{Section:motion planning}
Although the affine SEA-SCBF \eqref{eqn:affineCBF} appears simple, Boolean compositions (see, e.g., \cite{glotfelter2020nonsmooth}) allow it to model complex environments. This insight motivates us to propose a novel motion planning framework.

For comparison purposes, we extend the expectation-based discrete-time CBF (ED) \cite{cosner2023robust,cosner2024bounding} and the probabilistic CBF (PCBF) \cite{mestres2025probabilistic} into the formulations under state estimation uncertainty, respectively, and then compare them with the proposed SEA-SCBF-QP \eqref{eqn:SEASCBFQP} for safe motion planning (i.e., collision-free trajectory generation) in a dense environment.

The ED-based constraint is given by \cite{cosner2023robust,cosner2024bounding} 
\begin{equation}
\EC{h(x_{k+1})}{\mathcal{F}_k} 
\geq \alpha  h(x_{k}),
\label{eqn:ED}
\end{equation}
where $\alpha \in (0,1)$. Incorporating state uncertainty and estimation, we extend \eqref{eqn:ED} into 
\begin{equation}
\EC{h(x_{k+1})}{\mathcal{F}_k} 
\geq \alpha  \EC{h(x_{k})}{\mathcal{F}_k},
\label{eqn:SEA-ED}
\end{equation}
where $\alpha \in (0,1]$. For linear system \eqref{eqn:LTI} with affine CBF \eqref{eqn:affineCBF}, \eqref{eqn:SEA-ED} becomes $c^\top B\, u_k
\geq m(\mu_k)$, where $m(\mu_k)$ is given by \eqref{eqn:m(mu)}. Then, the corresponding SEA-ED-QP can be formulated as
\begin{equation}
\begin{aligned}
& \argmin_{{u}_k \in \mathcal{U}} 
& & \|u_k - \bar{u}_k\|^2 \\
& \quad\;
\textnormal{s.t.}
& &  
c^\top B \,u_k
\geq
m(\mu_k),
\end{aligned}
\label{eqn:SEAEDQP}
\end{equation}
in which $m(\mu_k)$ is given by \eqref{eqn:m(mu)}.

The PCBF-based chance constraint is given by \cite{mestres2025probabilistic}
\begin{equation}
\Prob{
h(x_{k+1}) \geq \alpha h(x_k)
}
\geq
\delta,
\label{eqn:PCBF}
\end{equation}
where $\alpha \in [0,1]$, $\delta \in (0,1)$. Incorporating state uncertainty and estimation, we extend \eqref{eqn:PCBF} into
\begin{equation}
\Prob{
h(x_{k+1}) \geq \alpha h(x_k) 
\mid \mathcal{F}_k
}
\geq
\delta,
\label{SEA-PCBF}
\end{equation}
where $\alpha \in [0,1]$, $\delta \in (0.5,1)$. Denote $\Delta h_k = h(x_{k+1}) - \alpha h(x_k)$, for \eqref{eqn:LTI} with \eqref{eqn:affineCBF}, the mean and variance of $\Delta h_k$ conditioned on $\mathcal{F}_k$ can be calculated as
\begin{equation*}
\EC{\Delta h_k}{\mathcal{F}_k} 
=
c^\top B u_k 
+ 
c^\top(A-\alpha I)\mu_k - (1-\alpha)b,
\end{equation*}
and 
\begin{equation*}
\VarC{\Delta h_k}{\mathcal{F}_k} 
=
c^\top
\left(
(A-\alpha I)\Sigma_k(A-\alpha I)^\top + \Sigma_{\varepsilon_k}
\right)c,
\end{equation*}
respectively. Then, \eqref{SEA-PCBF} is equivalent to
\begin{equation}
c^\top B\, u_k
\geq
m(\mu_k)
+
s(\Sigma_k, \Sigma_{\varepsilon_k}), \notag
\end{equation}
in which $m(\mu_k)$ is given by \eqref{eqn:m(mu)}, and 
\begin{equation}
s(\Sigma_k, \Sigma_{\varepsilon_k})
= 
\phi^{-1}(\delta)
\sqrt{
c^\top \!
\left(
(A-\alpha I)\Sigma_k (A-\alpha I)^\top
\!+
\Sigma_{\varepsilon_k}
\right)
c},
\label{eqn:s(sig)}
\end{equation}
where $\phi^{-1}$ is the inverse of the standard Gaussian cumulative distribution function $\phi(q) = \frac{1}{\sqrt{2\pi}}\int_{-\infty}^q \! \Exp{-\frac{1}{2}q^2} \mathrm{d}q$, which is why the lower bound of $\delta$ in \eqref{SEA-PCBF} becomes $0.5$ instead of $0$. Hence, the corresponding SEA-PCBF-QP can be formulated as
\begin{equation}
\begin{aligned}
& \argmin_{{u}_k \in \mathcal{U}} 
& & \|u_k - \bar{u}_k\|^2 \\
& \quad\;
\textnormal{s.t.}
& &  
c^\top B \,u_k
\geq
m(\mu_k)
+
s(\Sigma_k, \Sigma_{\varepsilon_k}).
\end{aligned}
\label{eqn:SEAPCBFQP}
\end{equation}

\begin{table}[t]
\centering
\caption{Comparison of Different Methods for Motion Planning under Different Settings (with 500 MC Trials).}
\label{tab:comparison}
\setlength{\tabcolsep}{3pt}
\renewcommand{\arraystretch}{1.15}
\begin{tabular}{ll l
S[table-format=3.1]
S[table-format=3.1]
S[table-format=3.1]}
\toprule
{Env} & 
{} & 
{Metric} & 
\textbf{SEA-SCBF} & {SEA-ED} & 
{SEA-PCBF} \\
\midrule
\multirow{2}{*}{\text{Accurate}}
 &  & Safety Rate ($\%$) & 99.0 & 69.0 & 95.8 \\
 &  & Goal Reach ($\%$)
 & 100.0 & 100.0 & 12.6 \\
\midrule
\multirow{2}{*}{\text{Inaccurate}}
 &  & Safety Rate ($\%$) & 95.6 & 62.2 & 75.8 \\
 &  & Goal Reach ($\%$)   & 100.0 & 100.0 & 9.6 \\
\bottomrule
\end{tabular}
\end{table}
We now apply the methods \eqref{eqn:SEASCBFQP}, \eqref{eqn:SEAEDQP}, and \eqref{eqn:SEAPCBFQP} for motion planing (trajectory generation) in a corridor with densely distributed obstacles. Specifically, for all three methods, we consider \eqref{eqn:LTI} with  $p_k \in \R^3$, $A = I$, $B = I\Delta t$, $H = I$, $\varepsilon_k \sim \mathcal{N}(0,0.06^2I)$, $\epsilon_k \sim \mathcal{N}(0,0.2^2I)$, $\alpha = 0.9$, $T = 240$. The nominal controller is a go-to-position controller $\bar{u}_k = p_{\text{g}} - \mu_k$ with the goal position $p_{\text{g}} = (12,0,0)^\top$ and the initial position $p_0 = (0,0,0)^\top$. In addition, we let $\phi^{-1}(\delta) = 3.93$ for the SEA-PCBF-QP \eqref{eqn:SEAPCBFQP}, and $\beta = \phi^{-1}(\delta)$ in $\beta_k = \beta \Exp{-7 \widetilde{Y}_k}$ for the SEA-SCBF-QP \eqref{eqn:SEASCBFQP}. We set 136 affine CBFs with $h^{\text{obs}}_{ij}(p) = (c^{\text{obs}}_{ij})^\top p - b^{\text{obs}}_{ij}$, $\forall i \in \{1,\dots,11\}$, $\forall j \in \{1,\dots,12\}$, representing 11 obstacles (regular dodecahedrons), and $h^{\text{wall}}_k(p) = (c^{\text{wall}}_k)^\top p - b^{\text{wall}}_k$, $\forall k \in \{1,2,3,4\}$, representing 4 walls in the $y$ and $z$ directions. The 136 CBFs are composed through the Boolean operator AND (see, e.g., \cite{glotfelter2020nonsmooth}), resulting in $h(p) = \min \{h^{\text{obs}}_{ij}(p), h^{\text{wall}}_k(p)\}$ for all $i \in \{1,\dots,11\}$, $j \in \{1,\dots,12\}$, and $ k \in \{1,2,3,4\}$, so there is only one active CBF at each time step. 
\begin{figure}[t]
\centering
\subfigure[SEA-SCBF]{
\begin{minipage}[b]{0.4\textwidth}
\includegraphics[width=1\textwidth]{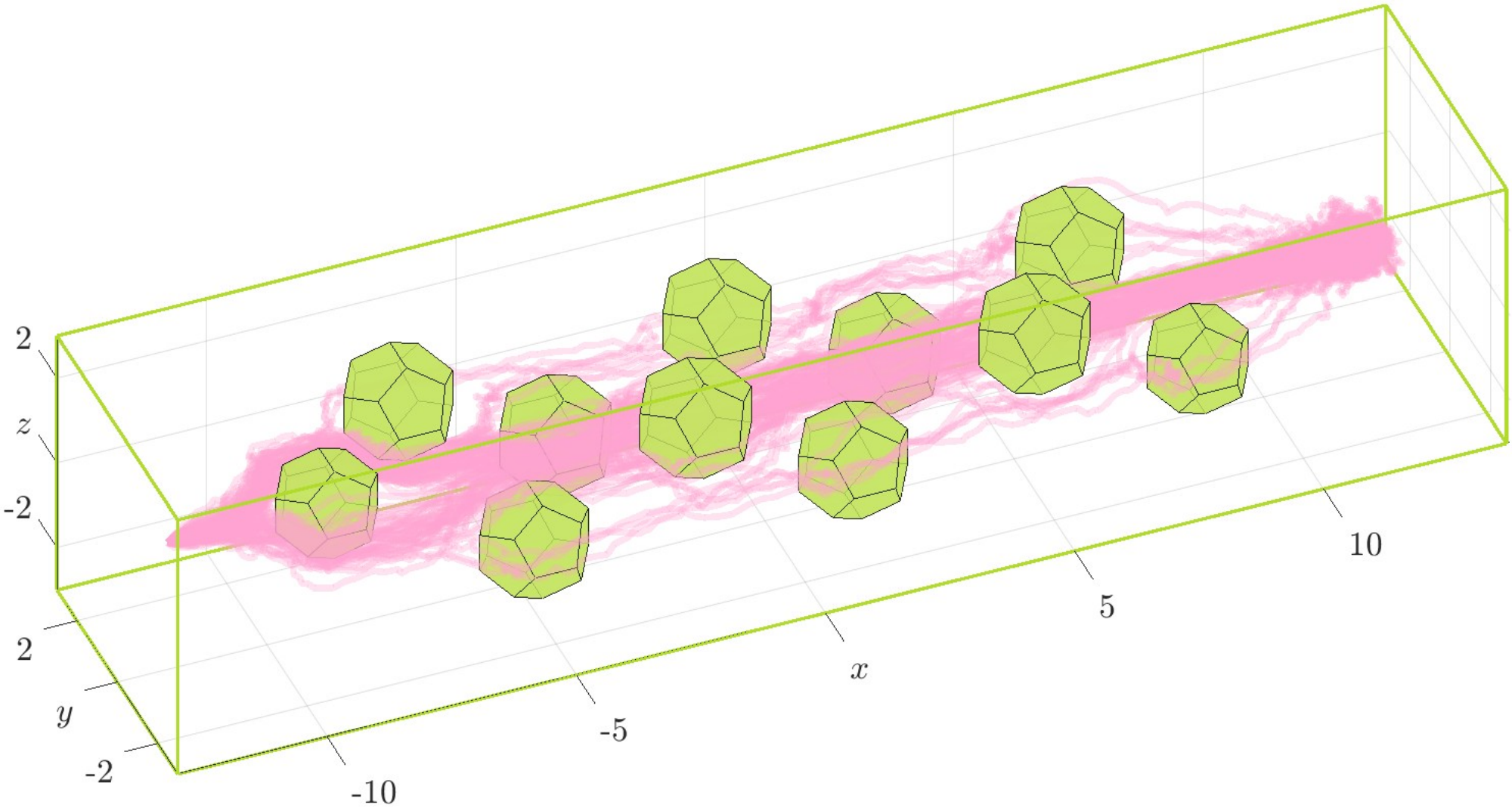}
\end{minipage}
}
\subfigure[SEA-ED]{
\begin{minipage}[b]{0.4\textwidth}
\includegraphics[width=1\textwidth]{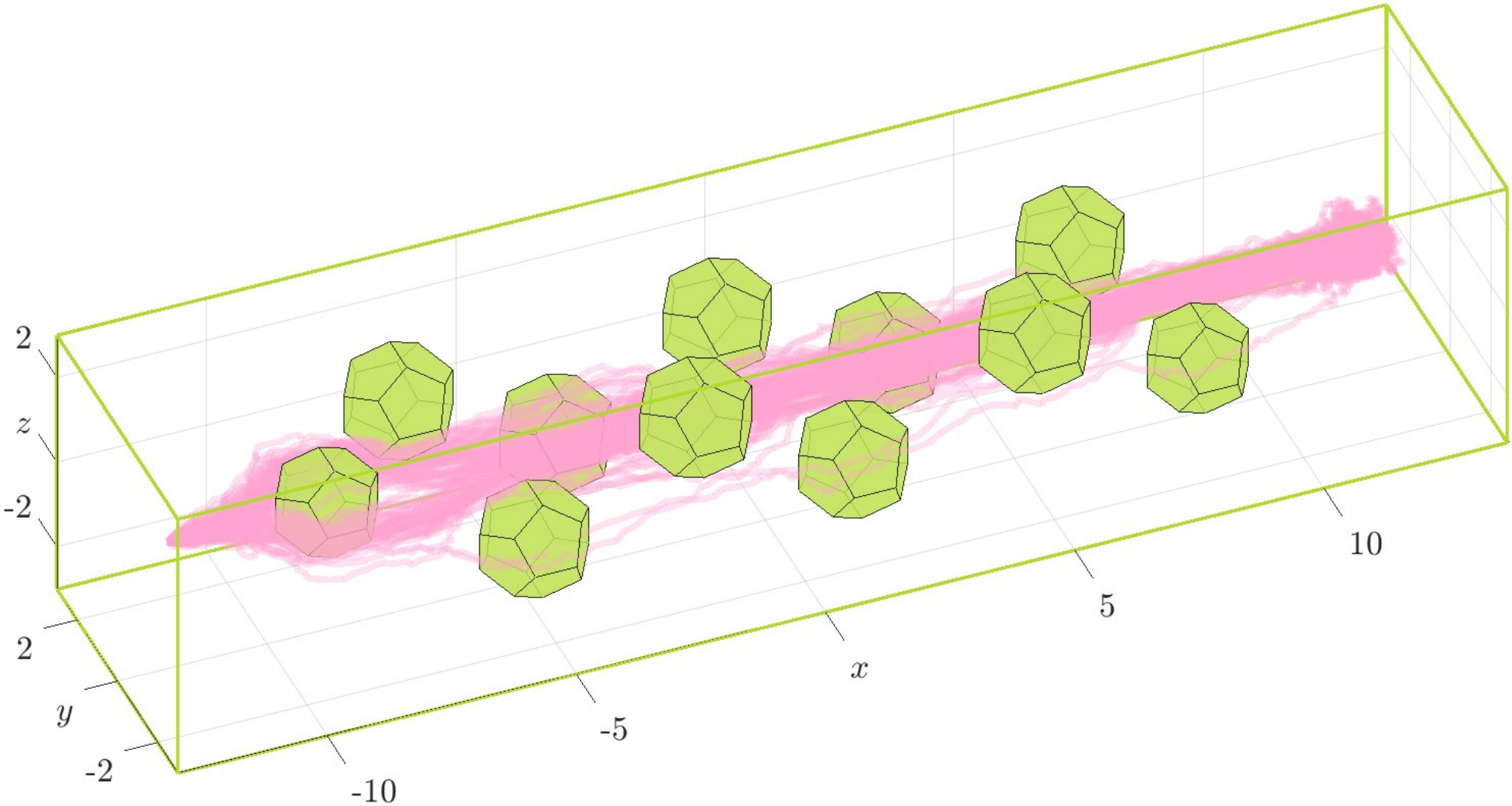}
\end{minipage}
}
\subfigure[SEA-PCBF]{
\begin{minipage}[b]{0.4\textwidth}
\includegraphics[width=1\textwidth]{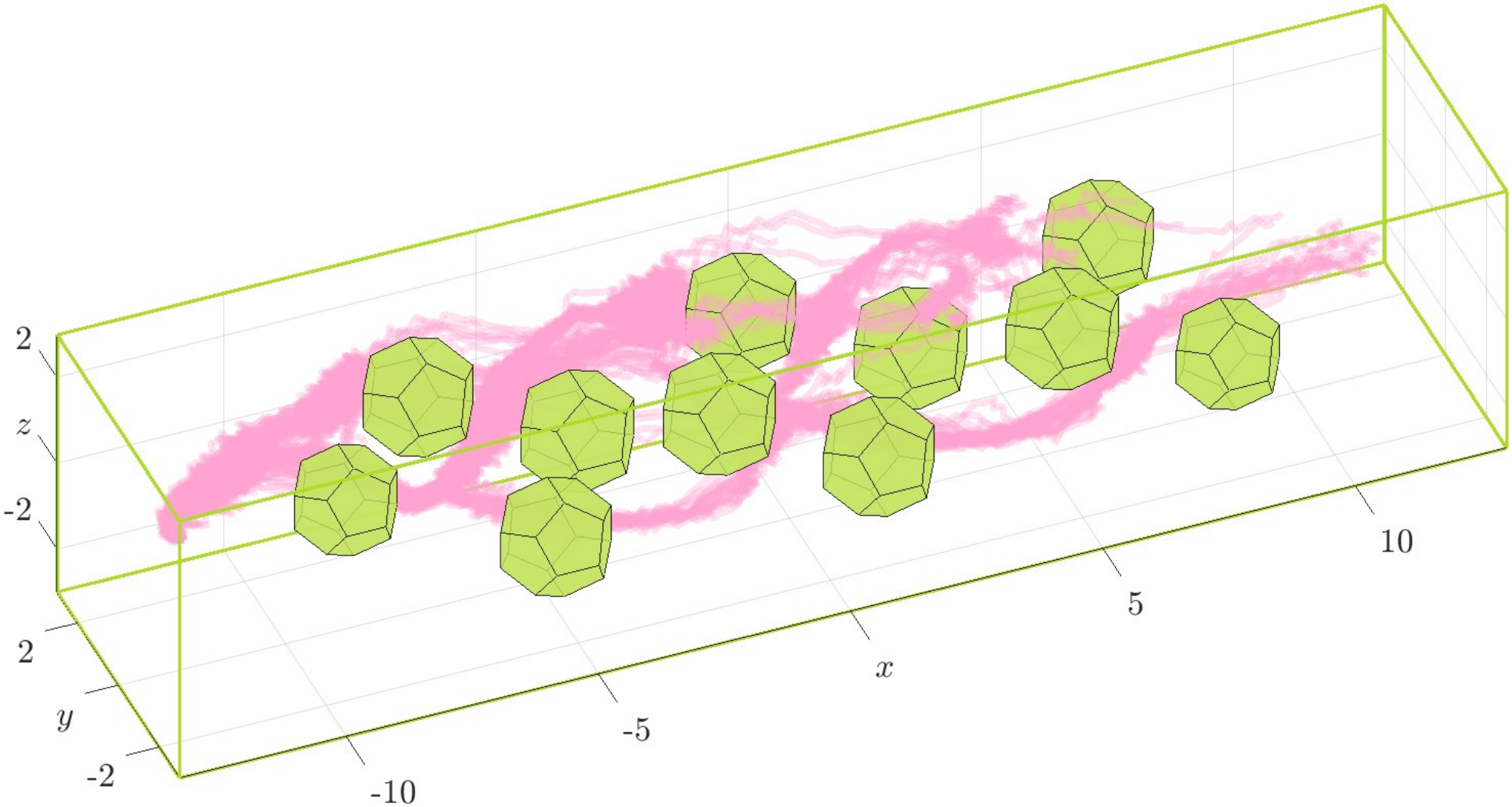}
\end{minipage}
}
\caption{Visualization of 100 trajectories (pink) among 500 MC trials with accurate obstacle (green) information generated by the SEA-SCBF-QP \eqref{eqn:SEASCBFQP}, the SEA-ED-QP \eqref{eqn:SEAEDQP}, and the SEA-PCBF-QP \eqref{eqn:SEAPCBFQP}, respecitvely.
}
\label{fig:MP}
\end{figure}
The simulation results (with 500 MC trials) are shown in Table~\ref{tab:comparison}, in which \textit{Safety Rate} denotes the percentage of collision-free trajectories, \textit{Goal Reach} denotes the percentage of trajectories that reach the goal point, and \textit{Env} indicates the accuracy of the closest obstacle facet used for active CBF selection, with \textit{Accurate} corresponding to ground-truth closest distances and \textit{Inaccurate} corresponding to estimated closest distances. 100 trajectories from the 500 MC trials with accurate environmental information are visualized in Fig.~\ref{fig:MP}.

As shown in Table~\ref{tab:comparison}, under accurate environmental information, the SEA-SCBF-QP \eqref{eqn:SEASCBFQP} achieves higher safety rate than the SEA-ED-QP \eqref{eqn:SEAEDQP}. This is because there is a term $\rho(\Sigma_k, \Sigma_{\varepsilon_k})$ in \eqref{eqn:SEASCBFQP}, leading to the  adaptivity to the level of uncertainty, i.e., with a worse estimator or higher noise, the controller tends to behave more conservatively, whereas there is no such a term in \eqref{eqn:SEAEDQP}. Another adaptivity of the the SEA-SCBF-QP \eqref{eqn:SEASCBFQP} is on how much the actions react to the level of uncertainty, as seen from the comparison between $\rho(\Sigma_k, \Sigma_{\varepsilon_k})$ in \eqref{eqn:rho(sig)} and $s(\Sigma_k, \Sigma_{\varepsilon_k})$ in \eqref{eqn:s(sig)}, i.e., when the estimated barrier function value is smaller, the actions resulting from \eqref{eqn:SEASCBFQP} become more conservative, whereas the ones from \eqref{eqn:SEAPCBFQP} keep the same level of reaction to the uncertainty no matter what the estimated barrier function value is. This can lead to overly conservative behavior. As a result, most candidate trajectories are obstructed by the obstacles, which is why the goal reaching rate achieved by the SEA-PCBF-QP \eqref{eqn:SEAPCBFQP} is significantly lower than the other two methods. 

A notable fact is that even though the SEA-PCBF-QP \eqref{eqn:SEAPCBFQP} typically results in more conservative actions, its safety rate is lower than the SEA-SCBF-QP \eqref{eqn:SEASCBFQP}. This is because when the estimated barrier function value is small, i.e., the estimated state is closer to an obstacle facet, the influence of the level of state estimation uncertainty has a greater impact in \eqref{eqn:SEASCBFQP} than in \eqref{eqn:SEAPCBFQP}. Especially when the level of the measurement noise is relatively higher than that of the process noise, the effect of the state estimation uncertainty is suppressed in \eqref{eqn:s(sig)}, resulting in relatively aggressive actions from \eqref{eqn:SEAPCBFQP} even when the state estimation uncertainty is high.  

In addition, we empirically evaluate the three methods under mapping uncertainty. In robotic systems incorporating simultaneous localization and mapping (SLAM), measurements used for state estimation are obtained from state-related sensors, such as IMU and visual odometry, while environment-related sensors, such as LiDAR and cameras, provide geometric information about surrounding obstacles. Although the mapping uncertainty is not the main focus of this paper, we evaluate the proposed methods under imperfect environment-related sensing, with the active CBF selected based on the estimated closest obstacle facet. As shown in Table~\ref{tab:comparison}, under inaccurate environmental information, each metric for each method is not higher (and typically slightly lower) than the one obtained with accurate environmental information. Moreover, under inaccurate environmental information, the relative performance trends across all methods remain consistent with those with accurate environmental information.

\section{SEA-SCBF on Lie Groups} \label{Section:LieGroup}
As discussed in Section~\ref{Section:Introduction}, for the vast majority of robotic systems, the states of interest evolve on manifolds rather than in Euclidean space. This distinction becomes particularly important in stochastic settings, where noise is not additive as it typically is in Euclidean space. In this section, we focus on the special Euclidean group which is a unimodular Lie group and arises in robotic systems commonly encountered in practice, with the state corresponding to the pose of a robot. In this case, the binary operation $\circ$ is matrix multiplication and the identity $\mathcal{I}$ is the identity matrix $I$ with the corresponding dimension. The intuition of left-invariance is a description of motion or error as seen in the body frame, remaining invariant to the choice of world frame, which can be mapped to the right-invariant form via the group adjoint action in Definition~\ref{def:adjoint}.

First, we review the concepts in Section~\ref{Section:LieGrouptool} specifically for $\mathrm{SE}(3)$ (see, e.g., \cite{murray2017mathematical,lynch2017modern,chirikjian2016harmonic}). The special Euclidean group $\mathrm{SE}(3)$ is the semi-direct product of the special orthogonal group $\mathrm{SO}(3)$ with the Abelian group $(\R^3,+)$, i.e., $\mathrm{SE}(3) = 
\mathrm{SO}(3) \ltimes (\mathbb{R}^3,+)$, where 
\begin{equation*}
\mathrm{SO}(3) 
=
\left\{
R \in \R^{3 \times 3} 
\mid R^\top R
=I, \, \mathrm{det}(R) = +1
\right\}.
\end{equation*}
A point $g \in \mathrm{SE}(3)$ takes the form of
\begin{equation*}
g = 
\begin{pmatrix} 
R & p \\ 
0^\top & 1 
\end{pmatrix},
\end{equation*}
in which $R \in \mathrm{SO}(3)$, $p \in \R^3$. 

The Lie algebra of $\mathrm{SO}(3)$ is
\begin{equation*}
\mathfrak{so}(3)
=
\left\{
\Omega \in \mathbb{R}^{3 \times 3} 
\mid 
\Omega = -\Omega^T
\right\},
\end{equation*}
and the Lie algebra of $\mathrm{SE}(3)$ is
\begin{equation*}
\mathfrak{se}(3)
=
\left\{
\begin{pmatrix} 
\Omega & v \\ 
0^\top & 0 
\end{pmatrix}
\,\middle|\,
\Omega \in \mathfrak{so}(3), \,v \in \R^3
\right\}.
\end{equation*}
Any element $\Omega \in \mathfrak{so}(3)$ can be represented by an angular velocity
$\omega = (\omega_1,\omega_2,\omega_3)^\top \in \R^3$ via
\begin{equation*}
\Omega = \omega^\wedge
=
\begin{pmatrix}
0 & -\omega_3 & \omega_2\\
\omega_3 & 0 & -\omega_1\\
-\omega_2 & \omega_1 & 0
\end{pmatrix} \in \mathfrak{so}(3),
\end{equation*}
and any element $\xi \in \mathfrak{se}(3)$ can be represented by a twist (i.e., an infinitesimal screw motion) $\xi^\vee = (\omega,v)^\top \in \R^6$. Accordingly, the left-invariant vector field induced by $\xi$ is
\begin{equation*}
\Gamma_\xi (g) 
= 
\begin{pmatrix} 
R\,\omega^\wedge & R\,v \\ 
0^\top & 0 
\end{pmatrix}.
\end{equation*}

Let $e_1,e_2,e_3$ be the standard basis of $\R^3$, then the standard basis of $\mathfrak{se}(3)$ is given by
\begin{equation*}
E_i
=
\begin{pmatrix} 
e_i^\wedge & 0 \\ 
0^\top & 0 
\end{pmatrix}, \quad
E_{i+3}
=
\begin{pmatrix} 
0 & e_i \\ 
0^\top & 0 
\end{pmatrix}, \quad \forall i = 1,2,3.
\end{equation*}
The group adjoint of $g \in \mathrm{SE}(3)$ acting on $\xi \in \mathfrak{se}(3)$ can be calculated as
\begin{equation*}
\left(g\,\xi\, g^{-1}\right)^\vee
= 
\left[\mathrm{Ad}_g \right] 
\,\xi^\vee,
\end{equation*}
in which
\begin{equation*}
\left[\mathrm{Ad}_g \right] 
=
\begin{pmatrix} 
R & 0 \\ 
p^\wedge R & R 
\end{pmatrix}.
\end{equation*}
Let $\theta = \|\omega\|$, then for any $\xi^\vee = (\omega,v)^\top$,
\begin{equation*}
\mathrm{Exp}(\xi) 
= 
\begin{pmatrix} 
\mathrm{Exp}(\omega^\wedge) & J(\omega)v \\ 
0^\top & 1 
\end{pmatrix},
\end{equation*}
in which
\begin{equation*}
\mathrm{Exp}(\omega^\wedge)
=
I + \frac{\sin(\theta)}{\theta}\omega^\wedge
+
\frac{1-\cos(\theta)}{\theta^2} (\omega^\wedge)^2, \quad \theta \neq 0,
\end{equation*}
\begin{equation*}
J(\omega)
=
I + 
\frac{1-\cos(\theta)}{\theta^2}\omega^\wedge
+
\frac{\theta - \sin(\theta)}{\theta^3} (\omega^\wedge)^2, \quad \theta \neq 0,
\end{equation*}
and $\mathrm{Exp}(\omega^\wedge) = J(\omega) = I$ when $\theta = 0$.

Consider the stochastic system on $\mathrm{SE}(3)$
\begin{equation}
\begin{cases}
g_{k+1} = g_k \,U_k \, \mathrm{Exp}(\varepsilon_k^\wedge), \\
z_k = H(g_k, \epsilon_k),
\end{cases}
\label{eqn:SEdynamics}
\end{equation}
in which $g_k, U_k \in \mathrm{SE}(3)$, $H: \mathrm{SE}(3) \times \R^l \to \mathcal{Z}$, $\varepsilon_k \sim \mathcal{N}(0, \Sigma_{\varepsilon_k})$, $\epsilon_k \sim \mathcal{N}(0, \Sigma_{\epsilon_k})$, 
$\Sigma_{\varepsilon_k} \in \mathbb{S}^6_+$, and $\Sigma_{\epsilon_k} \in \mathbb{S}^l_+$. The measurement manifold $\mathcal{Z}$ can be $\mathrm{SE}(3)$ itself or simply $\R^3$.

Then, we review state estimation with respect to \eqref{eqn:SEdynamics} as follows (see, e.g., \cite{chirikjian2016harmonic,wolfe2011bayesian,hartley2020contact,barfoot2024state}, for more detailed discussions on information fusion and state estimation on Lie groups). By using Lemma~\ref{lemma:BCH}, the prediction and update steps can be obtained as follows.

Prediction:
\begin{align*}
\mu_{k+1}^- 
&= \mu_k \, U_k, \\
\Sigma_{k+1}^- 
&= 
\big[\mathrm{Ad}_{U_k^{-1}}\big] \Sigma_k
\big[\mathrm{Ad}_{U_k^{-1}}\big]^\top
+ \Sigma_{\varepsilon_k}.
\end{align*}

Update:
\begin{align*}
\mu_{k+1} 
&= \mu_{k+1}^- \,\mathrm{Exp}((K_{k+1}\,r_{k+1})^\wedge),\\
\Sigma_{k+1}
&= (I - K_{k+1} H_{k+1})
\,\Sigma_{k+1}^-
(I - K_{k+1} H_{k+1})^\top \\
&\quad\,+
K_{k+1}\, \Sigma_{\epsilon_{k+1}} K_{k+1}^\top,
\end{align*}
in which
$K_{k+1} = \Sigma_{k+1}^- H_{k+1}^\top (H_{k+1} \Sigma_{k+1}^- H_{k+1}^\top + \Sigma_{\epsilon_{k+1}})^{-1}$. For pose measurement 
\begin{equation}
z_{k+1} = g_{k+1} \, \mathrm{Exp}(\epsilon_{k+1}^\wedge),\notag
\label{eqn:posemeasure}
\end{equation}
the innovation and the measurement matrix can be given by
\begin{equation*}
r_{k+1}
= \mathrm{Log}^\vee ((\mu_{k+1}^-)^{-1}z_{k+1}),\quad H_{k+1} = I,
\end{equation*}
respectively. For position-only measurement,
\begin{equation}
z_{k+1} =
p_{k+1} + \epsilon_{k+1}, \notag
\label{eqn:positionmeasure}
\end{equation}
the innovation and the measurement matrix can be given by
\begin{equation*}
r_{k+1} = z_{k+1} - p_{k+1}^-, \quad H_{k+1} = (0, R_{k+1}^-),
\end{equation*}
respectively, with
\begin{equation*}
\mu_{k+1}^-
=
\begin{pmatrix} 
R_{k+1}^- & p_{k+1}^- \\ 
0^\top & 1 
\end{pmatrix}.
\end{equation*}

Finally, we utilize the Taylor expansion per Lemma~\ref{lemma:Taylor} to approximate the SEA-SCBF constraint in \eqref{eqn:SE-CBF_Controller} on $\mathrm{SE}(3)$. Let $h\in C^2$ be an SEA-SCBF on $\mathrm{SE}(3)$, defining the safe set as
\begin{equation}
\mathcal{C} 
= 
\left\{ g\in \mathrm{SE}(3) 
\,\big|\,
h(g)\ge 0
\right\}.
\label{eqn:safesetSE3}
\end{equation}
Denote the left-invariant error by
$
\delta g = \mathrm{Log}^\vee (\mu^{-1} g).
$
Then, given $\mathcal{F}_k$, the SEA-SCBF $h$ at $g_k$ and $g_{k+1}$ can be approximated by
\begin{equation}
h(g_k) 
=
h(\mu_k)
+
\nabla_{\!\xi}h(\mu_k) \,\delta g_k
+
\frac{1}{2} {\delta g_k}^{\!\!\top}
\nabla^2_{\!\xi} h(\mu_k)\, \delta g_k,
\label{eqn:TaylorCBF1}
\end{equation}
and
\begin{align}
h(g_{k+1}) 
=
h(\mu_{k+1}^-)
&+
\nabla_{\!\xi}h(\mu_{k+1}^-) \,
{\delta g_{k+1}^-} \notag \\
&+
\frac{1}{2} (\delta g_{k+1}^-)^{\!\top}
\nabla^2_{\!\xi} h(\mu_{k+1}^-)\, {\delta g_{k+1}^-},
\label{eqn:TaylorCBF2}
\end{align}
neglecting the higher order terms, respectively, in which
\begin{equation*}
\nabla_{\!\xi} h
=
\left(
\mathcal{L}_{E_1}h, \dots, \mathcal{L}_{E_6}h
\right),
\end{equation*}
\begin{equation*}
\nabla^2_{\!\xi} h
=
\begin{pmatrix} 
\mathcal{L}_{E_1} \mathcal{L}_{E_1}h & \cdots & \mathcal{L}_{E_1} \mathcal{L}_{E_6}h\\ 
\vdots & \ddots & \vdots\\
\mathcal{L}_{E_6} \mathcal{L}_{E_1}h & \cdots & \mathcal{L}_{E_6} \mathcal{L}_{E_6}h
\end{pmatrix}.
\end{equation*}

Since $\EC{{\delta g_k}}{\mathcal{F}_k} = 0$ and $\EC{\delta g_k \delta g_k^\top}{\mathcal{F}_k} = \Sigma_k$ according to Definition~\ref{def:geometricmoment}, then taking expectations conditioned on $\mathcal{F}_k$ on \eqref{eqn:TaylorCBF1} yields
\begin{align}
\EC{h(g_k)}{\mathcal{F}_k} 
&= 
h(\mu_k) 
+ \frac{1}{2} \EC{\mathrm{tr}(\nabla^2_{\!\xi}h(\mu_k)\, \delta g_k \delta g_k^\top)}{\mathcal{F}_k}
\notag\\
&= 
h(\mu_k) 
+ 
\underbrace{
\frac{1}{2} \mathrm{tr}(\nabla^2_{\!\xi}h(\mu_k)\, \Sigma_k)
}_{\text{curvature correction term}}.
\label{eqn:curvature1}
\end{align} 
Since ${\delta g_{k+1}^-} = \mathrm{Log}^\vee ((\mu^-_{k+1})^{-1} g_{k+1}) = \big[\mathrm{Ad}_{U_k^{-1}}\big] \delta g_k + \varepsilon_k$ by using Lemma~\ref{lemma:BCH} and neglecting the higher order Lie bracket terms, then $\EC{{\delta g_{k+1}^-}}{\mathcal{F}_k} = 0$. Likewise, taking expectations conditioned on $\mathcal{F}_k$ on \eqref{eqn:TaylorCBF2} yields 
\begin{equation}
\EC{h(g_{k+1})}{\mathcal{F}_k} 
= 
h(\mu^-_{k+1}) 
+ 
\underbrace{
\frac{1}{2} \mathrm{tr}(\nabla^2_{\!\xi}h(\mu^-_{k+1})\, \Sigma^-_{k+1})
}_{\text{curvature correction term}}.
\label{eqn:curvature2}
\end{equation}

Note that the curvature correction terms in \eqref{eqn:curvature1} and \eqref{eqn:curvature2} remain for highly nonlinear SEA-SCBFs. The intuition for this is that by taking conditional expectations, the first-order terms are averaged out. In this scenario, if we do not incorporate the curvature correction terms, this is equivalent to treating purely linear functions, which cannot adequately approximate highly nonlinear SEA-SCBFs.

Moreover, taking variance conditioned on $\mathcal{F}_k$ on \eqref{eqn:TaylorCBF2} yields
\begin{align}
\VarC{h(g_{k+1})}{\mathcal{F}_k}
&=
\nabla_{\!\xi}h(\mu^-_{k+1})
\, \Sigma^-_{k+1} 
(\nabla_{\!\xi}h(\mu^-_{k+1}))^\top,
\label{eqn:VarSE3CBF}
\end{align}
neglecting the higher order terms.

In the following content, we apply the SEA-SCBF control synthesis framework \eqref{eqn:SE-CBF_Controller} together with \eqref{eqn:curvature1}, \eqref{eqn:curvature2}, and \eqref{eqn:VarSE3CBF} to $\mathrm{SE}(2)$ and $\mathrm{SE}(3)$, separately. Note that the corresponding expressions for $\mathrm{SE}(2)$ can be obtained as lower-dimensional special cases of $\mathrm{SE}(3)$.

\subsection{SEA-SCBF on $\mathrm{SE}(2)$} \label{section:SE2}
\begin{figure}[t]
\centering
\subfigure[]{
\begin{minipage}[b]{0.22\textwidth}
\includegraphics[width=1\textwidth]{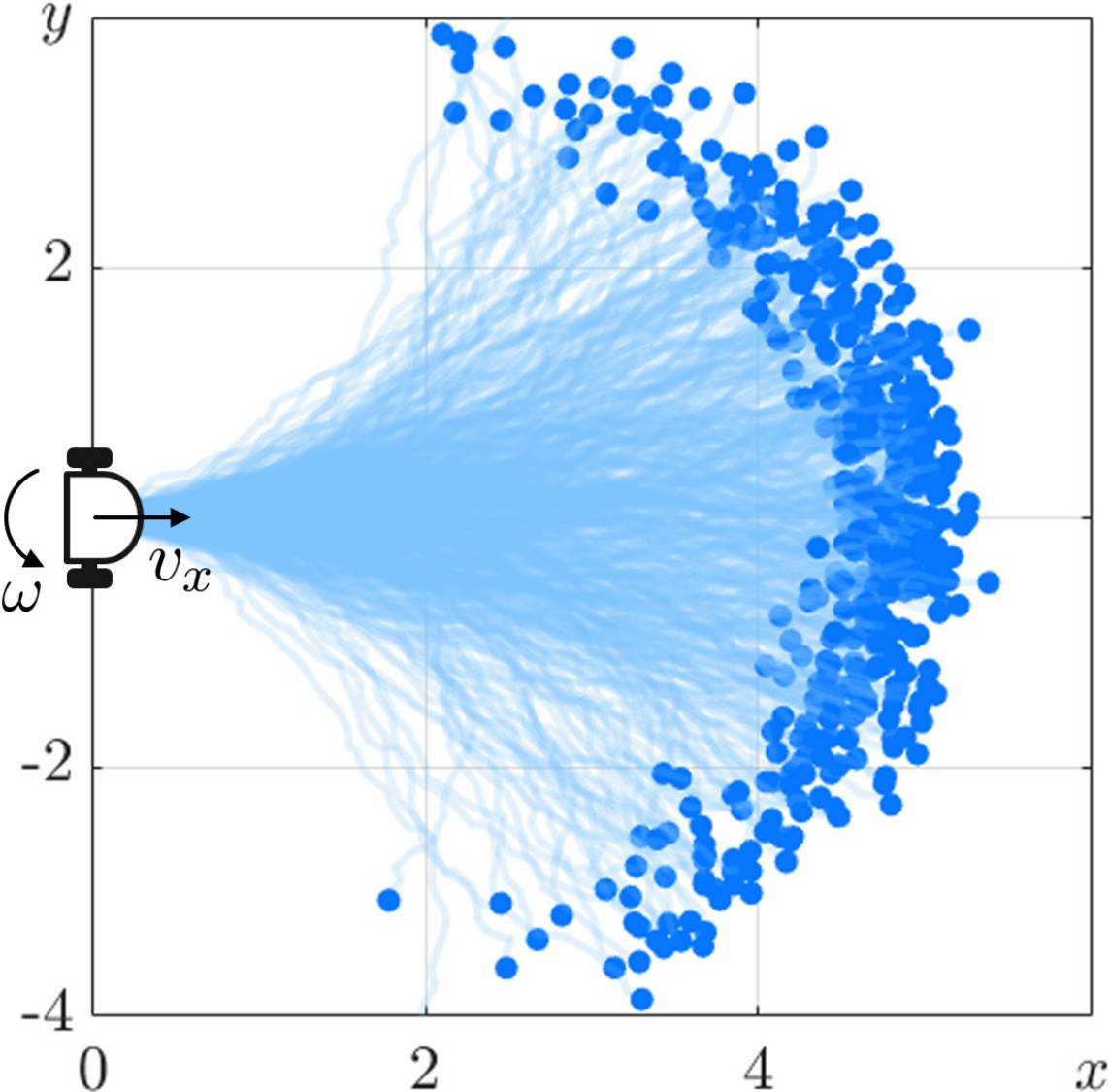}
\end{minipage}
\label{fig:SE2}
}
\subfigure[]{
\begin{minipage}[b]{0.216\textwidth}
\includegraphics[width=1\textwidth]{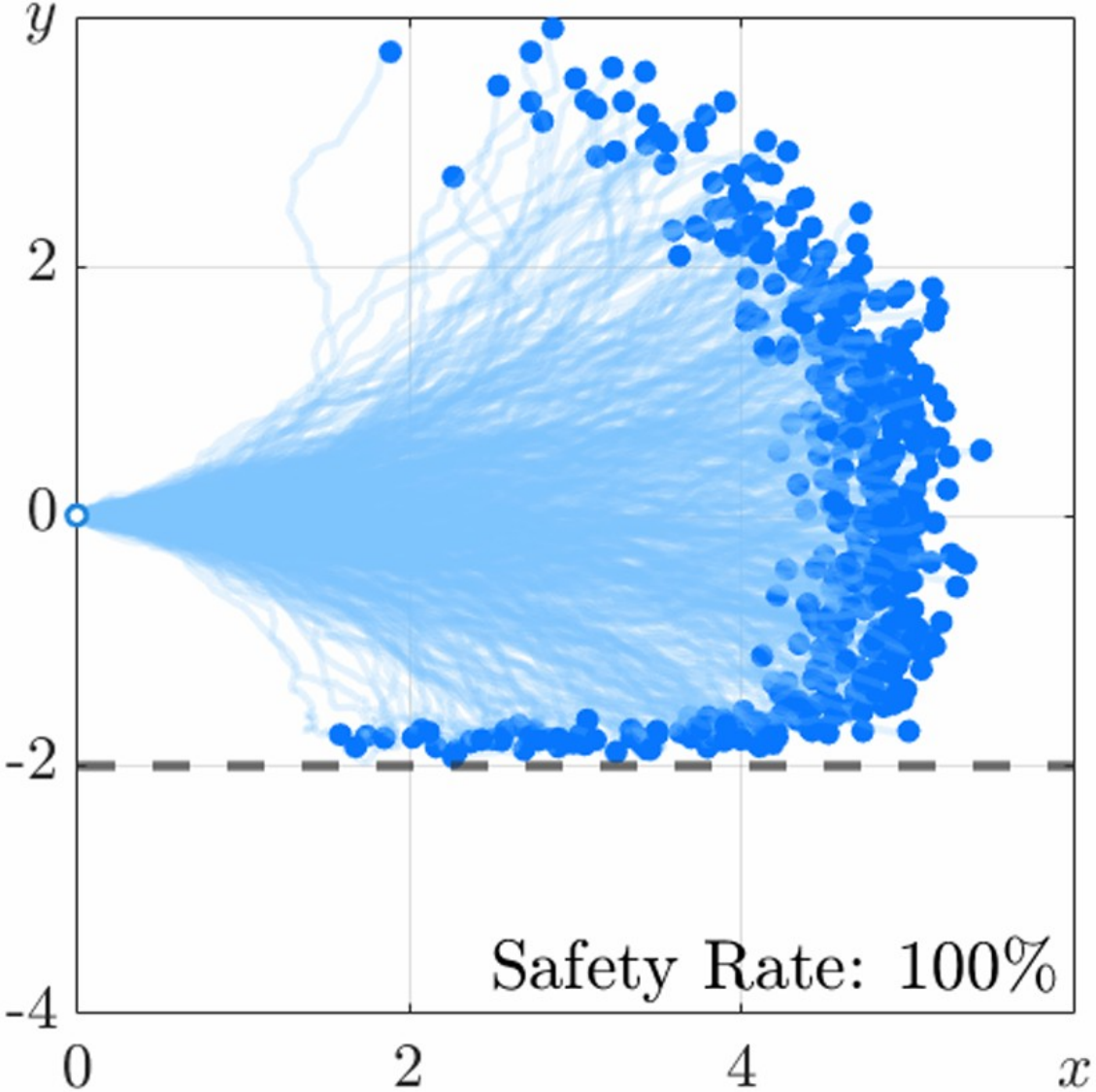}
\end{minipage}
\label{fig:SE2safe}
}
\caption{500 MC trials of a nonholonomic differential-drive wheeled robot with stochastic dynamics on $\mathrm{SE}(2)$.
The light blue lines represent the MC trajectories and the blue dots represent the final positions. (a) Without safety-filter; (b) With the safety-filter \eqref{eqn:opt_SE2}. The safety rate denotes the percentage of the trajectories remaining in the safe set among all trials.
}
\label{fig:SE2CBF}
\end{figure}
In this subsection, we apply the proposed framework \eqref{eqn:SE-CBF_Controller} to a nonholonomic differential-drive wheeled robot as illustrated in Fig.~\ref{fig:SE2}, whose state evolves on $\mathrm{SE}(2)$. Specifically, we consider \eqref{eqn:SEdynamics} with $U_k = \mathrm{Exp}(\xi_k \Delta t) \in \mathrm{SE}(2)$, $\varepsilon_k \sim \mathcal{N}(0,\diag(0.1^2, 0.03^2, 0.03^2))$, $\alpha = 0.9$, $\beta_k = 2 \Exp{-8 \widetilde{Y}_k}$, and $T = 50$. In terms of the body velocity $\xi_k = (\omega_k, v_k)^\top \in \mathfrak{se}(2)$, since its a nonholonomic robot, the lateral ($y$-axis of body-frame) component of the linear velocity is constrained to be zero, and we let the nominal control input be an open-loop controller driving the robot to move straight ahead, i.e., $\bar{\xi}_k = (0,1,0)^\top$. The SEA-SCBF is set to be 
\begin{equation}
h(g) = e_2^\top g \,e_3 +2,
\label{eqn:SE2CBF}
\end{equation}
and thus the curvature correction terms in \eqref{eqn:curvature1} and \eqref{eqn:curvature2} are neglected. Since the safety requirement does not concern the pose of the robot, we consider position-only measurement $z_k = p_k + \epsilon_k$ with $\epsilon_k \sim \mathcal{N}(0,\diag(0.05^2, 0.05^2))$.

As such, the stochastic safety-critical controller under state estimation uncertainty can be synthesized through the optimization problem
\begin{equation}
\begin{aligned}
& \argmin_{{\xi}_k \in \mathcal{U} \subset \R^3} 
& & \|\xi_k - \bar{\xi}_k\|^2 \\
& \quad\;\,
\textnormal{s.t.}
& &  
h(\mu_k U_k)
- 
\beta_k \sqrt{\Upsilon(\xi_k)} \geq 
\alpha h(\mu_k),
\end{aligned}
\label{eqn:opt_SE2}
\end{equation}
in which $\Upsilon(\xi_k) 
=
\nabla_{\!\xi}h(\mu_k U_k)
\, \Sigma^-_{k+1} 
(\nabla_{\!\xi}h(\mu_k U_k))^\top$. 

The simulation results (with 500 MC trials) are shown in Fig.~\ref{fig:SE2CBF}. Under the nominal open-loop controller that drives the robot along the $x$-axis of its body frame, the final positions from 500 MC trials exhibit a banana distribution \cite{thrun2000real,long2013banana}. With the safety filter \eqref{eqn:opt_SE2}, all MC trajectories stay in the safe set 
$\mathcal{C} = \left\{ g\in \mathrm{SE}(2) 
\,\big|\,
h(g)\ge 0
\right\}$ characterized by \eqref{eqn:SE2CBF}, which demonstrates the effectiveness of the proposed SEA-SCBF control synthesis framework \eqref{eqn:SE-CBF_Controller} on $\mathrm{SE}(2)$.

\subsection{SEA-SCBF on $\mathrm{SE}(3)$} \label{section:SE3}
\begin{figure}[t]
\centering
\includegraphics[scale=0.4]{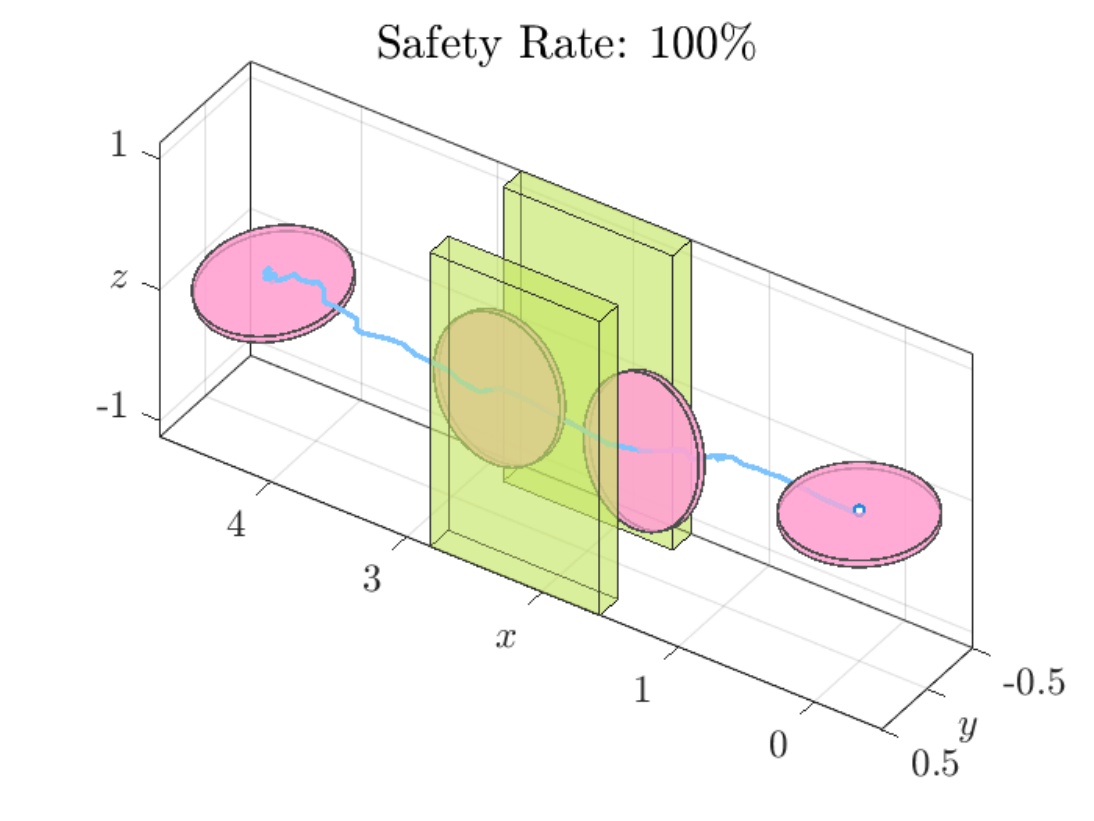}
\caption{Snapshots of the rigid body (pink) with stochastic dynamics on $\mathrm{SE}(3)$ navigating through the vertical slit (green) at four different time steps along one trajectory. The safety rate denotes the percentage of collision-free trajectories among 500 MC trials.}
\label{fig:SE3}
\end{figure}
As discussed in Section~\ref{Section:Introduction}, when the size of a robot is comparable to the scale of the environment, especially in crowded environments, the pose of the robot must be taken into account. In this scenario, safety can no longer be described solely by its position, such as a rigid body navigating through a slit considered in this subsection. Specifically, the rigid body is set to be a disk with radius $r_{\text{d}} = 0.5$, and the safe set \eqref{eqn:safesetSE3} is characterized by
\begin{equation}
h(g)
=
a_{\text{s}}(1-\chi_{\text{s}}(p)) + \chi_{\text{s}}(p) h_{\text{s}}(g),
\label{eqn:CBF_slit}
\end{equation}
where  $\chi_{\text{s}}(p) = \mathrm{exp}
\big(\!
-\!\frac{1}{2}\|p-c_{\text{s}}\|^2_{\Sigma_{\text{s}}^{-1}}\big)
$, $c_{\text{s}} = \frac{1}{2}(c_{\text{s},1} + c_{\text{s},2})+d_{\text{s}}$, $h_{\text{s}}(g) = -\frac{1}{\kappa_{\text{s}}} \mathrm{log}(\mathrm{exp}(-\kappa_{\text{s}} \varphi_1)+\mathrm{exp}(\kappa_{\text{s}} \varphi_2))$, $\varphi_i(g) = n_{\text{s}}^\top (p - c_{\text{s},i}) - r_{\text{d}}\sqrt{1-(n_{\text{s}}^\top R n_{\text{r}})^2}- d_{\text{s}}$, $\forall i = 1,2$, $\Sigma_{\text{s}}\succ 0$, $d_{\text{s}}>0$, $d_{\text{s}}>0$, $a_{\text{s}}>0$,  $\kappa_{\text{s}} >0$, $c_{\text{s},i}$ is the center of each wall, $n_{\text{s}}$ is the unit normal of the slit in the world frame, and $n_{\text{r}}$ is the unit normal of the rigid body in the body frame (see \cite{letti2025safety} for more details of the slit barrier function \eqref{eqn:CBF_slit} construction).

We consider \eqref{eqn:SEdynamics} with $U_k = \mathrm{Exp}(\xi_k \Delta t) \in \mathrm{SE}(3)$. Since the safety requirement concerns the pose of
the robot, we utilize the pose measurement $z_{k+1} = g_{k+1} \, \mathrm{Exp}(\epsilon_{k+1}^\wedge)$. The process noise and the measurement noise are $\varepsilon_k \sim \mathcal{N}(0,\diag(0.08^2 I_3, 0.03^2 I_3))$ and $\epsilon_k \sim \mathcal{N}(0,0.05^2 I_6)$, respectively, where $I_d$ denotes the identity matrix with dimension $d$. In addition, $\alpha = 0.9$, $\beta_k = 2 \Exp{-8 \widetilde{Y}_k}$, and $T = 70$. The nominal control input is a closed-loop go-to-pose controller driving the rigid body toward a goal pose with the same orientation as the initial pose, i.e., $\bar{\xi}_k = \mathrm{Log}^\vee (\mu_k^{-1} g_{\text{g}})$ with $p_{\text{g}} = (4.5,0,0)^\top$, $p_0 = (0,0,0)^\top$, and $R_{\text{g}} = R_0 = I$. 
\begin{figure}[t]
\centering
\includegraphics[scale=0.32]{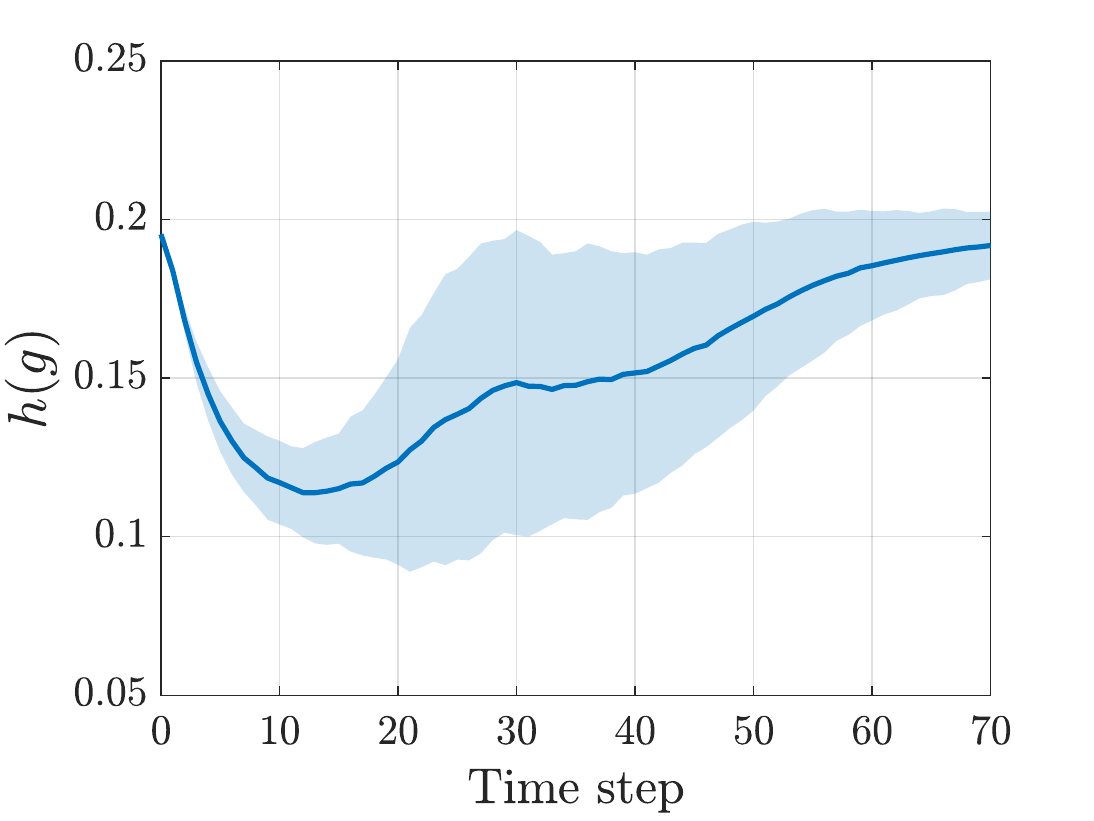}
\caption{SEA-SCBF value over time steps. The solid line indicates the mean over 500 MC trials, and the shaded region denotes $\pm$1 standard deviation.}
\label{fig:SE3_hg}
\end{figure}
As such, the stochastic safety-critical controller under state estimation unceratainty can be synthesized via the optimization problem
\begin{equation}
\begin{aligned}
& \argmin_{{\xi}_k \in \mathcal{U} \subset \R^6} 
& & \|\xi_k - \bar{\xi}_k\|^2 \\
& \quad\;\,
\textnormal{s.t.}
& &  
\Psi(\xi_k) 
- \beta_k \sqrt{\Upsilon(\xi_k)} \geq 
\alpha \Xi_k,
\end{aligned}
\label{eqn:opt_SE3}
\end{equation}
in which 
\begin{align*}
\Psi(\xi_k) 
&=
h(\mu_k U_k) 
+
\frac{1}{2} \mathrm{tr}(\nabla^2_{\!\xi}h(\mu_k U_k)\, \Sigma^-_{k+1}),\\
\Upsilon(\xi_k) 
&=
\nabla_{\!\xi}h(\mu_k U_k)
\, \Sigma^-_{k+1} 
(\nabla_{\!\xi}h(\mu_k U_k))^\top,\\
\Xi_k 
&= 
h(\mu_k) 
+
\frac{1}{2} \mathrm{tr}(\nabla^2_{\!\xi}h(\mu_k)\, \Sigma_k).
\end{align*}
Note that unlike the application in Section~\ref{Section:motion planning} where there are $136$ barrier functions composed through Boolean operators, here we only have a single one \eqref{eqn:CBF_slit}, although the price to pay is that it is highly nonlinear, which is why the curvature correction terms in \eqref{eqn:curvature1} and \eqref{eqn:curvature2} remain in \eqref{eqn:opt_SE3}.

The simulation results (with 500 MC trials) are shown in Fig.~\ref{fig:SE3} and Fig.~\ref{fig:SE3_hg}. The nominal closed-loop controller attempts to keep the rigid body translating while maintaining the initial orientation. When a slit appears ahead, the safety filter \eqref{eqn:opt_SE3} enables the rigid body to adaptively reorient itself while passing through the slit in a safe (i.e., collision-free) manner. All MC trajectories stay in the safe set $\mathcal{C} = \left\{ g\in \mathrm{SE}(3) 
\,\big|\,
h(g)\ge 0
\right\}$ characterized by \eqref{eqn:CBF_slit}, i.e., the barrier function values of all MC trajectories stay nonnegative for all time, which demonstrates the effectiveness of the SEA-SCBF control synthesis framework \eqref{eqn:SE-CBF_Controller} on $\mathrm{SE}(3)$.

\section{Conclusion} \label{Section:Conclusion}
In this work, we propose a safety-critical control framework for stochastic systems when safety requirements are defined on the true state but only noisy state information is available. The proposed framework incorporates uncertainty arising jointly from dynamics and measurement noise, enables probabilistic reasoning about finite-time safety, and supports online control synthesis that adapts to the level of uncertainty. In linear settings, closed-form expressions are derived, resulting in a quadratic program for control synthesis. This further suggests that stochasticity, rather than being merely a source of uncertainty, can be actively exploited in motion planning, offering a new perspective on trajectory generation. The effectiveness of the proposed framework is also demonstrated on systems evolving on $\mathrm{SE}(2)$ and $\mathrm{SE}(3)$, separately, highlighting its applicability beyond Euclidean state spaces.

\bibliographystyle{IEEEtran}
\bibliography{references}{}
\end{document}